\documentclass{article}
\usepackage{hyperref, graphicx}
\usepackage{amsthm}
\usepackage{amsmath}
\usepackage{tocloft}
\usepackage{amssymb}
\usepackage{slashed} 
\usepackage{yfonts}
\usepackage[T1]{fontenc}
\usepackage[utf8]{inputenc}
\usepackage[english]{babel}
\usepackage{tikz-cd}
\usepackage{mathabx}
\usepackage{stmaryrd}
\usepackage{braket}
\usepackage{physics}
\usepackage{geometry}
\usepackage[
backend=biber,
style=numeric-comp,
sorting=none,
]{biblatex}
\theoremstyle{theorem}
\newtheorem{theorem}{Theorem}[section]
\newtheorem{lemma}{Lemma}[section]
\newtheorem{corollary}{Corollary}[section]
\newtheorem{proposition}{Proposition}[section]

\theoremstyle{remark}
\newtheorem{remark}{Remark}[section]

\theoremstyle{example}

\theoremstyle{definition}
\newtheorem{definition}{Definition}[section]

\usepackage{mathtools}

\usepackage{mathrsfs}
\numberwithin{equation}{section}

\newcommand{\beq}{\begin{eqnarray}}
\newcommand{\eeq}{\end{eqnarray}}

\hypersetup{
pdfstartview = {FitH},
}
\hypersetup{
	colorlinks=true,         
	linkcolor=brown,          
	citecolor=red,        
	urlcolor=blue            
}

\usepackage{xcolor}

\newcommand*{\colorboxed}{}
\def\colorboxed#1#{%
  \colorboxedAux{#1}%
}

\newcommand*{\colorboxedAux}[3]{%
  \begingroup
    \colorlet{cb@saved}{.}%
    \color#1{#2}%
    \boxed{%
      \color{cb@saved}%
      #3%
    }%
  \endgroup
}

\def\d{\mathfrak d}
\def\g{\mathfrak g}

\def\r{\mathfrak r}

\def\D{\mathfrak{D}}

\makeatletter
\newsavebox{\@brx}
\newcommand{\llangle}[1][]{\savebox{\@brx}{\(\m@th{#1\langle}\)}%
  \mathopen{\copy\@brx\kern-0.5\wd\@brx\usebox{\@brx}}}
\newcommand{\rrangle}[1][]{\savebox{\@brx}{\(\m@th{#1\rangle}\)}%
  \mathclose{\copy\@brx\kern-0.5\wd\@brx\usebox{\@brx}}}
\makeatother

\def\R{\mathbb{R}}
\def\bbC{\mathbb{C}}

\def\id{\mathrm{id}}

\newcommand{\cB}{{\cal B}}
\newcommand{\cJ}{{\cal J}}

\newcommand{\ovo}{{_{({-1})}}}

\usepackage{authblk}
  \allowdisplaybreaks
  \usepackage{atbegshi}
\AtBeginDocument{\AtBeginShipoutNext{\AtBeginShipoutDiscard}\addtocounter{page}{-1}}
\begin{document}

\title{Integrability from categorification and Kac-Moody 2-algebras}

\author[1,2]{{ \sf Hank Chen}}\thanks{chunhaochen@bimsa.cn, c398chen@uwaterloo.ca}
\author[1]{{\sf Florian Girelli}}\thanks{fgirelli@uwaterloo.ca}

\affil[1]{\small Department of Applied Mathematics, University of Waterloo, 200 University Avenue West, Waterloo, Ontario, Canada, N2L 3G1}

\affil[2]{\small Beijing Institute of Mathematical Sciences and Applications, Beijing 101408, China}

\maketitle

\begin{abstract}
     The theory of Lie bialgebras and the classical Yang-Baxter equation plays a major role in the study of 1+1d integrable systems; many families of integrable systems can be recovered from a Lax pair which is constructed from a Lie bialgebra. A categorified, higher homotopy notion of Lie algebras has been studied, which gave rise to the notion of (strict) Lie 2-bialgebras  (Lie algebra crossed-modules) and the 2-graded Yang-Baxter equations. In this paper, we use these differential graded structures to generalize the construction of a Lax pair and introduce an appropriate notion of higher-dimensional integrability. Within this framework, we introduce a higher derived version of the affine Kac-Moody algebra, which underpins the 2-graded Lax integrability that we have developed here as a zero 2-curvature condition. As an explicit demonstration, we will consider a 3d field theory and show that it (i) is 2-graded Lax integrable and (ii) hosts symmetries governed by a Kac-Moody 2-algebra. 
\end{abstract}

\tableofcontents

\newpage

\section{Introduction}
It is well-known since the 19th century that physical classical systems can be described by a {\it Poisson structure} on the space $M$ of configurations, which is a bilinear skew-symmetric map $\{-,-\}$ called the {\it Poisson bracket} on the space $C^\infty(M)$ of smooth functions on $M$ satisfying the Leibniz rule and the Jacobi identity. 
The classical dynamics, under a chosen {\it Hamiltonian} function $H\in C^\infty(M)$, of the observables $f\in C^\infty(M)$ in the system is governed by the Poisson bracket on $M$, through the following differential equation 
\begin{equation}
    \dot f = \{H,f\},\nonumber
\end{equation}
with some appropriate initial conditions.  The Poisson structure also serves as the precursor to canonical Dirac quantization, in the sense that $\{-,-\}$ contributes to the quantum commutator $[-,-]$ to first order in the quantum deformation parameter (usually denoted by $\hbar$), hence Poisson geometry in general plays a central role in the understanding of various physical systems, both classical and quantum.

A manifold $M$ equipped with a Poisson structure $\{-,-\}$ together with the choice of a Hamiltonian $(M,\{-,-\}$, $H)$ is called a {\it Hamiltonian system}. A particularly special class of Hamiltonian systems has "enough" conserved quantities --- namely functions $f$ with $\dot f=0$ --- such that the dynamics of $H$ can be separately described completely on the level sets of these conserved quantities. More precisely, we call $(M,\{-,-\},H)$ {\it (completely) integrable} if there exist $n$-number of constants of motion $f_1,f_2,\dots,f_n$ in involution, namely $\{f_i,f_j\}=0$ for all $1\leq i,j\leq n$. If the Poisson structure is in fact non-degenerate, then this would imply $2n =\operatorname{dim}M$. This is because such a system $\{f_i\}_i$ of conserved quantities partitions the Poisson manifold into leaves that are invariant under the dynamics generated by $H$.

Finding constants of motion is not an easy task, and integrable systems in general are difficult to characterize. To get a handle on them, the theory of {\it Lax pairs} had been developed \cite{Lax:1968fm,Bobenko:1989,book-integrable}. A Lax pair $(L,P)$ is a tuple of maps $L,P:M\rightarrow\mathfrak{g}$ into a Lie algebra $\mathfrak{g}$ such that, on an integral curve $\gamma:[0,1]\to M$ generated by the vector field $X_H= \{H,-\}$, the {\bf Lax equation} (here $t$ is the coordinate on $[0,1]$)
\begin{equation} 
    \partial_t({\gamma^*L}) = \gamma^* \{H,L\} = [\gamma^*L,\gamma^*P]
\end{equation}
holds, where $\gamma^*:C^\infty(M)\to C^\infty([0,1])$ is the pullback. If the Hamiltonian system $(M,\{-,-\},H)$ admits a Lax pair,  this Lax equation implies that the following \textit{trace polynomials} on any representation space $V$ of $\g$
\begin{equation}
    f_{V,m} = \operatorname{tr}_V \rho(L)^m,\qquad \rho: \g\rightarrow \mathfrak{gl}(V)\nonumber
\end{equation}
are constants of motion \cite{Meusburger:2021cxe} for any $m$, which are in involution provided the correct Poisson bracket is suitably defined (see eg. \eqref{dualpois} later). 

In particular, the eigenvalues of $L$ are conserved quantities, hence for $\g$ of sufficiently large rank, $(M,\{-,-\},H)$ can be made integrable \cite{book-integrable}. By the trace-determinant relation $\exp \circ \operatorname{tr} = \operatorname{det}\circ \exp$, these trace polynomials $f_{V,k}$ can be summed into the {\it monodromy determinants} \cite{book-integrable,Babelon_1992,Adler:1978,book-cft}
\begin{equation}
    M_V = \operatorname{det}_V\tilde\rho(e^L) = \sum_{k=0}^\infty \frac{1}{k!}f_{V,k}\label{mondrom}
\end{equation}
labeled by (irreducible) representations $\tilde\rho: G\rightarrow\operatorname{GL}(V)$ of $G$ corresponding to $\rho$. 

Finding a Lax pair on a general Hamiltonian system, if one even exists, is itself a difficult problem. However, if we are given a Poisson-Lie group $(G,\Pi)$, its corresponding bialgebra $\g$ \cite{Semenov1992} and   a Poisson map $\cJ:M\rightarrow \g^*$, then we can pull back the Poisson map $\cJ$, namely perform a particular change of canonical coordinates, to  achieve a Lax pair for the induced Hamiltonian system on $M$.

Indeed, from the bialgebra $\g$, we can infer a Poisson structure $\{-,-\}^*$  on the function algebra $C^\infty(\g^*)$ of the dual $\g^*$ Lie algebra for which a Lax pair $(L,P)$ can be {\it canonically} constructed for any invariant Hamiltonian $H\in C^\infty(\g^*)$, see for example a nice summary of the construction  in  \cite{Meusburger:2021cxe}. 

Many physical integrable systems, such as the classical Toda, Korteweig-de Vries and the Kadomtsev–Petviashvili hierarchies \cite{book-integrable}, as well as the XXX/XXZ/XYZ family of quantum spin chains \cite{Zhang:1991}, can be transformed in this way to the canonical integrable system on $\g^*$ for certain Lie bialgebras $\g$. Moreover, 1+1d integrable field theories such as the Wess-Zumino-Witten/principal chiral models and their various integrable deformations \cite{KLIMCIK1995455,KLIMCIK1996116,Hoare:2021dix,Mohammedi:2020qok} are well-known to also host Lax pairs that generically have singular dependence on the \textit{spectral parameter} $z\in\bbC$. However, these are all 1d or 1+1d models, and it is generally a difficult task to identify the notion of integrability for higher dimensional systems; see for example for some proposals \cite{Maillet:1989gg,Larsson:2002pk}. 

It is expected that increasing the  dimensionality of a system corresponds to a "categorification" of the relevant structure; this is the "categorical ladder = dimensional ladder" proposal \cite{Crane:1994ty, Baez:1995xq, Mackaay:ek,Pfeiffer2007}.   Kapranov and Voevodsky used this proposal to categorify the notion of vector space in order to find higher dimensional version  of the Yang-Baxter equations \cite{Kapranov:1994}, called the Zamolodchikov tetrahedron equations. While not providing obvious concrete integrable models, the Kapranov-Voevodsky approach has been very influential in the study of higher categorical structures (see for instance \cite{neuchl1997representation,Baez:2012,GURSKI20114225,KongTianZhou:2020,Johnson_Freyd_2023,Delcamp:2021szr,Delcamp:2023kew,decoppet2022morita} for a short list of developments). 

Apart from the Kapranov-Voevodsky approach, other routes toward categorification include {\it Soergel bimodules} in Khovanov's knot categorification \cite{Khovanov:2000,rouquier:hal-00002981,Rouquier2005CategorificationOS}, as well as the Baez-Crans higher Lie algebras that appear as vertical categorifications of Lie algebras \cite{Baez:2003fs,Baez:2004in,baez2004,Nikolaus_2014}. These higher Lie algebras --- or more precisely $L_n$-algebras and $L_\infty$-algebras \cite{Stasheff:1963} --- can be understood as derived/higher homotopy generalizations of Lie algebras, and they help to climb the dimensional ladder by defining a higher-dimensional notion of global \cite{Benini:2018reh,Cordova:2018cvg} and gauge \cite{Baez:2004in,Wockel2008Principal2A} symmetries. 

\medskip 

Higher-gauge theories based on $L_n$-algebras had also recently been used to investigate higher-dimensional integrable structures \cite{Schenkel:2024dcd,Chen:2024axr}, in which Lax higher-connections were found for various field theories. The current paper investigates the underlying classical integrability framework for such theories in 3-dimensions. The connection of this framework with solutions of the tetrahedron equations (see \textit{Remark \ref{zamTE}}) is, unfortunately, still unknown. This is the subject of a future work. 

\subsection*{Outline and main results}
The paper is organized as follows. In \S \ref{poisbialg}, we review some key aspects of the theory of Lie (1-)bialgebras and (1-)Lax integrability, following \cite{Meusburger:2021cxe,book-integrable}. We demonstrate how the Kac-Moody algebra $\widehat{\g_k}$ can be used to write the Lax equation as a zero-curvature condition in 1+1d, and illustrate the integrable properties of the celebrated 2d Wess-Zumino-Witten conformal field theory in this framework. 

The rest of the paper is divided into two parts, with the former setting up the latter.
\begin{enumerate}
    \item {\bf First main result, \S \ref{2integra}:} after reviewing Lie 2-(bi)algebras, Poisson differential graded (dg) manifolds and the classical 2-Yang-Baxter equation  (2CYBE) in \S \ref{poislie2grp} following   \cite{Bai_2013,chen:2022,pantev2013shifted,pridham2018outline,Chen:2012gz}, we define a notion of {\it 2-graded Lax integrability}, and study its conserved quantities. We then explicitly construct a 2-graded \textit{Lax pair} from solutions of the 2CYBE. 
    \item \textbf{Second main result, \S \ref{zero2curv} and \S \ref{2lax2km}:} taking the de Rham complex $\Omega^\bullet(\Sigma)$ of a 2d oriented surface $\Sigma$, we explicitly construct the associated higher Kac-Moody algebra "$\widehat{\Sigma_s\g}$" as a centrally-extended infinite-dimensional dg Lie algebra. This object is then used to write the 2-Lax equation defined in \S \ref{2integra} as a zero 2-curvature condition in 2+1d. As a concrete example, we establish the 2-Lax integrability of the 3d field theory $\mathcal{W}$ derived in \cite{Chen:2024axr}, and exhibit its $\widehat{\Sigma_s\g}$-symmetry. 
\end{enumerate}
The theory $\mathcal{W}$ can be understood as a higher-dimensional analogue of the {principal chiral model}.

We emphasize here that our construction in \S \ref{zero2curv} can be extended to take generic differential-graded commutative algebras (dgca's) $\mathcal{A}$ as input, not just $\Omega^\bullet(\Sigma)$. The choice of the dgca $\mathcal{A}$ depends on the higher-dimensional field theory under study; see \S \ref{2laxintheift} and \textit{Remark \ref{raviolo}}.

\medskip

Our conceptualization of higher current algebras fits with closely related work \cite{Gwilliam:2018lpo,Cirio:2012be,Alfonsi_2023,Garner:2023zqn} in the literature. There is also another notion named "2-Kac Moody algebra", which was studied by \cite{Rouquier20082KacMoodyA,Khovanov_2010} in relation to knot categorification. It is not obvious, though is a very interesting prospect to investigate, how derived current algebra considerations are related to knot homology considerations.

\subsubsection*{Statements and conflicts of interest declaration} 
There are no conflicts of interest to declare at this time. 

\subsubsection*{Data availability statement}
There are no data associated with this work.

\subsubsection*{Acknowledgments}
The authors would like to thank Joaquin Liniado, Nicolas Cresto and Benoit Vicedo for fruitful and enlightening discussions for the completion of this manuscript. HC is supported by the National Science Foundation of China (Grant Number: W2533012).

\section{Integrable systems from Lie bialgebras: a review}\label{poisbialg}
We will review first how to construct a Lax pair by picking a Lie algebra $\g$ and a classical $r$-matrix. We then recall the construction relevant to the Wess-Zumino-Witten model. 

\subsection{Lax pair}
In this section, we first recall how $r$-matrix associated to a Lie algebra $\g$ can be used to define a Lie algebra 1-cocycle, which upon integration will give rise to a Poisson bivector on the group $G$. We will then use this Poisson structure to construct the Lax pair.  

Let $\g$ denote a Lie algebra over a field $k$ of characteristic zero. Recall that a quasitriangular classical $r$-matrix $r\in \mathfrak{g}^{\otimes 2}$ on $\mathfrak{g}$ satisfies the {\it classical Yang-Baxter equation}
\begin{equation*}
    \llbracket r,r\rrbracket = [r_{12},r_{13}] +[r_{12},r_{23}]  + [r_{13},r_{23}]=0,
\end{equation*}
where $\llbracket-,-\rrbracket$ is the Schouten bracket (by treating $\g$ as left-invariant vector fields on $G$). Further, by splitting  $r= r^\wedge+r^\odot$ into skew-symmetric $r^\wedge\in\mathfrak{g}^{\wedge 2}$ and symmetric $r^\odot\in\mathfrak{g}^{\odot 2}$ components, we have
\begin{equation}
    \llbracket r^\wedge,r^\wedge\rrbracket = -\llbracket r^\odot,r^\odot\rrbracket\nonumber.
\end{equation}

The classical Yang-Baxter equation imposes the $\operatorname{ad}$-invariance of the symmetric part $r^\odot$, which makes it into a \textit{quadratic Casimir} --- namely an invariant symmetric quadratic form on $\mathfrak{g}$. Non-degenerate elements therein determine on $\mathfrak{g}$ a non-degenerate invariant symmetric bilinear form $ \langle-,-\rangle_{\odot}$, up to scalar factor. We identify the spaces $\g\cong\mathfrak{g}^*$ with this non-degenerate bilinear form by $X\mapsto g_X = \langle -,X\rangle_\odot\in\mathfrak{g}^*$ for each $X\in\mathfrak{g}$. For $\g$ semisimple, the pairing $\langle-,-\rangle_\odot$ can be fixed by choosing $r^\odot$ to be given by a  Killing form. 

\medskip 

On the other hand, the skew-symmetric part $r^\wedge$ defines a 1-cocycle $\psi \in Z^1(\mathfrak{g},\mathfrak{g}^{\wedge 2})$ given by
\begin{equation}
    \psi X = (\operatorname{ad}_X\otimes 1+ 1\otimes\operatorname{ad}_X)r^\wedge,\qquad X\in\g,\nonumber
\end{equation}
and induces the Lie bracket $[-,-]_*$ on $\mathfrak{g}^*$ defined with respect to the pairing form
\begin{equation}
    \langle [g,g']_*,X\rangle = \langle g\wedge g',\psi X\rangle, \quad g,g'\in\g^*.\nonumber
\end{equation}
Now for the pairing form $\langle-,-\rangle_\odot$ induced by a quadratic Casimir $r^\odot$, the skew-symmetric part $r^\wedge$ can be used to define a map (here  $\{T_i\}_i$ is a basis of $\mathfrak{g}$) $$\varphi:\mathfrak{g}\rightarrow\mathfrak{g},\qquad X\mapsto  (r^\wedge)^{ij}\langle X,T_i\rangle T_j\,.$$ If $r=r^\odot+r^\wedge$ satisfies the classical Yang-Baxter equations \eqref{cybe}, then the followin
\begin{align}
    [g_X,g'_X]_*& = \langle-,[\varphi(X),X'] + [X,\varphi(X')]\rangle_\odot,\qquad X,X'\in\mathfrak{g}\label{algmap}
\end{align}
for $g_X=\langle-,X\rangle\,,\,g'_X=\langle-,X'\rangle$ defines a Lie algebra structure on $\g^*$. This allows to identify yet another Lie bracket $[-,-]_r = [-,-]\circ(\varphi\otimes 1+1\otimes\varphi)$ on $\mathfrak{g}$, which we now use in \eqref{dualpois} to define a particular Poisson bracket on functions on $\g^*$.

\medskip

If we use the non-degenerate pairing $r^\odot$ to identify the cotangent spaces of $\mathfrak{g}^*$ with $\g$, we induce a Poisson bracket
\begin{equation}
    \{\phi,\phi'\}^*_r(g) = \langle g,[d_g\phi,d_g\phi']_r\rangle, \quad \phi,\phi'\in C^{\infty}(\g^*),\quad g\in \g^*\label{dualpois},
\end{equation}
called the \textit{Kirillov-Kostant $r$-bracket} on $C^\infty(\mathfrak{g}^*)$ \cite{Meusburger:2021cxe}, where the bracket $[-,-]_r$ is defined with the map $\varphi$ as in \eqref{algmap}. The quadratic Casimir $r^\odot$ is invariant under the coadjoint action of the Lie group $G$ on $\mathfrak{g}^*$,
\begin{equation}
    \langle\operatorname{Ad}^*_x g,X\rangle = \langle g,\operatorname{Ad}_{x^{-1}}X\rangle,\nonumber
\end{equation}
where $x\in G$ and $\operatorname{Ad}$ is the adjoint representation of $G$ on $\mathfrak{g}$.

This then allows us to construct invariant Hamiltonians $H$ on $\g^*$ from a non-degenerate $r^\odot$. This is significant, as the Hamiltonian system $(\mathfrak{g}^*,\{-,-\}^*,H)$ is automatically integrable (see \textbf{Theorem 4} of \cite{Meusburger:2021cxe}). Indeed, it admits a Lax pair $(L,P):\mathfrak{g}^*\rightarrow\mathfrak{g}$ satisfying 
\begin{equation}
    \partial_t(\gamma^*L)= \gamma^*\{H,L\}_r^* = [\gamma^*L,\gamma^*P]\label{lax},
\end{equation}
with  (recall repeated indices are summed) 
\begin{equation}
    L: g\mapsto (r^\odot)^{ij}\langle g,T_i\rangle T_j ,\qquad P:g\mapsto \varphi(d_gH) = (r^\wedge)^{ij}\langle T_i,d_gH\rangle T_j\label{laxLcond}
\end{equation}
in terms of the basis $\{T_i\}_i$ of $\mathfrak{g}$. Moreover, the Lax operator $L$ is "fundamental" in the sense that it also satisfies
\begin{equation}
    \{L,L\}^*_r = [L\otimes 1+ 1\otimes L,r^\wedge].\nonumber
\end{equation}
The trace polynomials of $L$ thus define conserved quantities, which are in involution with respect to the $r$-bracket thanks to the classical Yang-Baxter equations.

This construction of a canonical Lax pair on the Hamiltonian system $(\mathfrak{g}^*,\{-,-\}^*,H)$ is the key result that will be generalized in \S \ref{2integra}.

\subsection{Zero curvature representation of the Lax equation}\label{wzw}
An interesting fact regarding the Lax equation \eqref{lax} is that, in certain situations, it can be represented as a {zero curvature}
\begin{equation*}
    \dot A_x - \partial_x A_t + [A_t,A_x] = 0
\end{equation*}
for a connection $A = A_t dt+A_x dx$ in 1+1d. The easiest way to see this is to introduce the so-called \textbf{Kac-Moody current algebra} $\widehat{\g_k}$ of a simple Lie algebra $\g$ (over $\mathbb{C})$. This object is linearly isomorphic to the direct sum $\g((z))\oplus k\mathbb{C}$, where $\g((z)) = \g\otimes \bbC((z))\cong C_{\partial}(D^\times,\g)$ is the space of $\g$-valued formal Laurent series in $z$, or equivalently the space of $\g$-valued holomorphic functions on the formal punctured disc $D^\times$. 

\begin{remark}
    It would be useful for later to consider the Kac-Moody current algebra $\widehat{\g_k}\cong\widehat{\Omega_k\g}$ as a central extension of the \textit{base loop algebra} $\Omega\g=C(S^1,\g)$, namely the algebra of $\g$-valued functions on $S^1\subset \bbC$ \cite{Baez:2005sn,book-integrable}. The Laurent series $\bbC((z))$ are then realized as the Fourier modes on $S^1$ in this formulation. When we lift the Lax equation to the higher homotopy context in \S \ref{2lax2km}, this perspective connects the results of this paper with the derived current algebras that has appeared previously in the literature \cite{Garner:2023zqn} (see also \cite{FAONTE2019389,Kapranov2021InfinitedimensionalL}).
\end{remark}

We now review some basic facts about the Kac-Moody current algebra/affine Lie algebra $\widehat{\g_k} = \g((z))\oplus k\bbC$, following \cite{book-integrable,KAC197885,Baez:2005sn}. It has the following Lie algebra structure
\begin{align}
    [X\oplus \xi,X'\oplus \xi'] &= [X,X'] \oplus 2k\oint_{S^1}\langle X,\partial_zX'\rangle \nonumber\\ 
    &= \big([\mathsf{x},\mathsf{x}']\otimes (fg)(z)\big) \oplus 2k\langle \mathsf{x},\mathsf{x}'\rangle\operatorname{Res}_{z=0}(f\partial_z g),\label{kmalg}
\end{align}
where $X=\mathsf{x}\otimes f(z)$ and $X'=\mathsf{x}'\otimes g(z)$ are elements in $\g((z))$ and $\langle-,-\rangle$ denotes the Killing form on $\g$. In terms of the Laurent series, the residue can be understood as the Cauchy contour integral
\begin{equation*}
    \operatorname{Res}_{z=0}(f\partial_z g) = \frac{1}{2\pi i}\int \frac{dz}{z} f(z)\partial_z g(z)
\end{equation*}
over the formal punctured disc $D^\times$.\footnote{The loop algebra can be realized concretely in terms of $\g$-valued Fourier modes on the circle $S^1.$ To see this, we let $z=e^{i\theta}$ whence we can rewrite the residue $$\operatorname{Res}_{z=0}(f(z)z^n) = \frac{1}{2\pi }\oint_{S^1}d\theta f(\theta)e^{in\theta} = \hat f_n$$ in terms of the Fourier coefficient for square-integrable functions $f$ on $S^1$. This also makes it clear that the residue satisfies the following orthonormality relation $\operatorname{Res}_{z=0}z^{n+m} = \delta_{n+m,0}$.}

This structure can be understood as arising from a Lie algebra central extension 
\begin{equation*}
    k\mathbb{C}\rightarrow \widehat{\g_k}\rightarrow \g((z))
\end{equation*}
given by a Lie algebra 2-cocycle $k\in Z^2(\Omega\g,\mathbb{C})$ given by \cite{Baez:2005sn}
\begin{equation*}
    \mathsf{k}(X,X') = \frac{k}{2\pi i} \int\frac{  dz}{z}\langle X,\partial_zX'\rangle,
\end{equation*}
where $k\in\mathbb{C}^\times$ is now interpreted as a parameter called the \textit{level}. 

\begin{remark}\label{derivative}
    It is sometimes convenient to append an additional generator $D$ to the affine Kac-Moody algebra that implements the differentiation $[D,X] = \frac{d}{dx}X$ along the angular coordinate $x$ on $S^1$ \cite{book-integrable}, such that $\widehat{\g_k}\cong\Omega\g\oplus \lambda\bbC\oplus D\bbC$. This will not be necessary in this paper, however.
\end{remark}

The Kac-Moody algebra has also equipped a pairing induced from the Killing form,
\begin{equation}
    \langle X\oplus \xi,X'\oplus \xi'\rangle = \frac{1}{2\pi i}\oint_{S^1}dz \langle X,X'\rangle +\xi\xi',\label{kmpairing}
\end{equation}
from which the defining relation
\begin{equation*}
    \langle \operatorname{ad}_{X\oplus \xi}^*X''\oplus\xi'',X'\oplus \xi'\rangle =-\langle [X\oplus \xi,X'\oplus \xi'],X''\oplus\xi''\rangle
\end{equation*}
for the coadjoint action $\operatorname{ad}^*$ gives rise to the following expression
\begin{align}
    \operatorname{ad}^*_{X\oplus\xi}X'\oplus \xi' &= [X,X'] -i\xi'z\partial_z X \nonumber\\ 
    &= [\mathsf{x},\mathsf{x}']\otimes f(z)g(z) + \xi' \mathsf{x}\otimes\partial_xf, \label{coad1}
\end{align}
where $x\in S^1$ is the angle variable on the circle for which $-iz\partial_z = \partial_x$. 

Now suppose $(L\oplus 1,P\oplus 0)$ denotes a Lax pair valued in the Kac-Moody algebra $\widehat{\g_k}$. Taking $M=S^1\times\mathbb{R}$, for any integral curve $\gamma:[0,1]\to M$ generated by the timelike vector $\partial_t$ we have the Lax equation
\begin{equation*}
    \dot{L}=\partial_t L = \operatorname{ad}_P^*L\oplus 1 = [P,L] + \partial_xP
\end{equation*}
in accordance with \eqref{coad1}. If we now take $A_x = L, A_t = P$, then the Lax equation reads as
\begin{equation*}
    \partial_tA_x - \partial_xA_t = [A_t,A_x],
\end{equation*}
which is precisely the flatness equation for a $G$-connection on $S^1\times \mathbb{R}$.

Notice the splitting of the Laurent series into Taylor series in $z,z^{-1}$ is coisotropic with respect to the above pairing form,
\begin{equation}
    \langle -, -\rangle|_{\g[[z]]\otimes \g[[z]]} =0,\qquad \langle -,-\rangle|_{z^{-1}\g[[z^{-1}]]\otimes z^{-1}\g[[z^{-1}]],} =0.\label{kmcosio}
\end{equation}
Indeed, it is known that the Kac-Moody algebra is closely related to quantum groups \cite{Alekseev:1992wn,Tolstoy2002FromQA}, and in fact forms a Manin triple with respect to the cosiotropy splitting \eqref{kmcosio} \cite{Freidel:2019ees}. This is important, as the two linear projection maps $P_\pm: \mathfrak{d}\to \g_\pm$ of a Manin triple $\mathfrak{d} = (\g_+,\g_-,\langle-,-\rangle)$ give rise to solutions of the classical Yang-Baxter equation \cite{Semenov1992,Semenov-Tyan-Shanskii1983-ti}
\begin{equation}
        r = P_+ - P_-\,,\label{doubleCYB}
\end{equation}
and hence give rise to examples of classical integrable systems.

In \S \ref{zero2curv}, we will define the notion of a "Kac-Moody 2-algebra" suitable on 2d surfaces, such that the 2-graded Lax equation can be rewritten equivalently as a zero 2-curvature equation for a 2-connection \cite{Zucchini:2021bnn,Chen:2022hct}. 

\subsection{Lax formulation of the 2d Wess-Zumino-Witten model}
The Kac-Moody algebra plays a central role in the {\it Wess-Zumino-Witten model} \cite{WESS197195,WITTEN1983422,book-cft}. This is a 2d conformal sigma model, with the field $g:M^2\rightarrow G$ defined on a 2d oriented Riemannian manifold $M^2$ valued in a compact simple Lie group $G$. The action is given by
\begin{equation*}\label{actionW}
    W[g] = k\int_{M^2} \langle g^{-1}dg,\star_2 g^{-1}dg\rangle + \frac{k}{4}\int_{B^3}\langle \tilde g^{-1}d\tilde g,[\tilde g^{-1}d\tilde g,\tilde g^{-1}d\tilde g]\rangle,
\end{equation*}
where $\langle-,-\rangle$ is the canonical Killing form and $k$ is called the \textit{level} of the theory. The second term in $W$, called the \textit{Wess-Zumino term}, requires the extension $\tilde g:B^3\rightarrow G$ of $g$ to a 3d contractible manifold $B^3$ with $\partial B^3=M^2$. This term serves to quantize the level $\lambda \in \mathbb{Z}\cong \pi_3G$ \cite{Hoare:2021dix}.

\medskip

Let $(z,\bar z)\in M^2$ denote local complex coordinates on the worldsheet $M^2$. The equations of motion for the Wess-Zumino-Witten model implement the conservation
\begin{equation*}
    \partial_{\bar z}J=0,\qquad \partial_z \bar J=0
\end{equation*}
of the currents $J=k(\partial_z g)g^{-1},\bar J=k g^{-1}(\partial_{\bar z}g)$ given by the right and left Maurer-Cartan forms of $G$. Their conservation equations are equivalent due to the Maurer-Cartan identity
\begin{equation*}
    \partial_{\bar z} (\partial_zgg^{-1}) = \operatorname{Ad}_g \partial_z(g^{-1}\partial_{\bar z} g),
\end{equation*}
and they each define a flat $G$-connection on the worldsheet $M^2$. Hence they each give rise --- in a somewhat trivial manner --- to a Lax pair
\begin{equation*}
    (L,P) = (J_x,J_t),\qquad (\bar L,\bar P) = (\bar J_x,\bar J_t)
\end{equation*}
satisfying the Lax flatness integrability condition. 

Importantly, these Lax currents $J,\bar J$ individually satisfy Poisson brackets given precisely by the Kac-Moody algebra \cite{WITTEN1983422}. This can be derived from the two-sided global symmetry $\delta_\omega g\mapsto g.\omega - \bar\omega. g$ of the Wess-Zumino-Witten model, parameterized by holomorphic/anti-holomorphic symmetry parameters $\omega\in C_{\partial}(M^2,\g),\,\bar\omega\in C_{\bar\partial}(M^2,\g)$. The variational Poisson brackets of the currents \cite{KNIZHNIK198483}
\begin{align*}
    \delta_{\bar \omega}J &= \delta_\omega\bar J =0,\\ s
    \delta_\omega J &= [\omega,J] + k\partial_z \omega,\\
    \delta_{\bar \omega}\bar J&= [\bar\omega,\bar J] + k \partial_{\bar z}\bar\omega,
\end{align*}
give rise to two copies of the Kac-Moody current algebra, with $\lambda$ playing the role of the central charge. 

The infinite-dimensional nature of $\widehat{\g_k}$ then allows one to construct infinitely many monodromy charges (cf. \eqref{mondrom}) from the Lax pairs $(L,P),(\bar L,\bar P)$, labeled by irreducible representations $V,\bar V$ of the \textit{global} Kac-Moody group $\widehat{\Omega_kG}$ (see also \cite{KAC197885}).

The Wess-Zumino-Witten model serves as a simple and straightforward demonstration of some of the concepts involved in Lax integrability. In \S \ref{2lax2km}, we will exhibit the 2-Lax pairs in the 3d topological-holomorphic field theory $\mathcal{W}$ discovered in \cite{Chen:2024axr}, and show that it manifests the integrable properties analogous to those of the Wess-Zumino-Witten model discussed above.

\section{Strict Lie 2-bialgebras and Lie 2-algebra cocycles}\label{poislie2grp}
We now begin discourse on the categorified versions of the above. As such for the rest of this paper, the notation "$\g$" and "$G$" will from now on refer exclusively to Lie 2-algebras/Lie 2-groups.

\subsection{Lie 2-algebras and Lie algebra crossed-modules}
\begin{definition}\label{liexmod}
A Lie algebra crossed-module $(\mathfrak{g},\rhd,[-,-]_0)$ is the data of a pair of Lie algebras $(\mathfrak{g}_{-1}, [-,-]{\ovo})$, $(\mathfrak{g}_0, [-,-]_0)$, a Lie algebra derivation $\rhd:\mathfrak{g}_0 \rightarrow\operatorname{Der}\g_{-1}$ and a Lie algebra homomorphism $t:\g_{-1}\rightarrow \g_0$ (called the \textit{$t$-map}), satisfying the equivariance and Peiffer identities
\begin{equation}
    t(X\rhd Y) = [X,tY]_0,\qquad [Y,Y']{\ovo} = (tY)\rhd Y',\label{pfeif} 
\end{equation}
 for all $X,X',X''\in\mathfrak{g}_0,$ and $\forall Y, Y'\in\mathfrak{g}_{-1}$. 

We shall denote a Lie algebra crossed-module by $\g=(\mathfrak{g}_{-1}\xrightarrow{t}\mathfrak{g}_0,\rhd,[-,-]_0)$ \cite{chen:2022}. 
\end{definition}
\noindent The following are automatic from the above definition,
\begin{eqnarray}
[X,[X',X'']_0]_0+[X',[X'',X]_0]_0+[X'',[X,X']_0]_0=0,\nonumber \\
X\rhd (X' \rhd Y) - X' \rhd (X\rhd Y) - [X,X']_0\rhd Y = 0\label{2jacob}, 
\end{eqnarray}
which one may call "2-Jacobi identities".

On the other hand, it is well-known that Lie algebra crossed-modules are equivalent to {\it strict} $L_2$-algebras \cite{Bai_2013,Wagemann+2021} (the same is true at the group level \cite{Porst2008Strict2A}). This is an equivalence of categories.
\begin{definition}\label{2algebra}
A {\it strict} $L_2$-algebra is a strict 2-term $L_\infty$-algebra. Explicitly, one has a graded space $\g\cong V_{-1}\oplus V_0$ equipped with $n$-ary operations $\mu_n\in\operatorname{Hom}^{2-n}(\g^{n\wedge},\g)$ given by
\begin{equation*}
    n=1:\quad \mu_1:V_{-1}\rightarrow V_0,\qquad n=2: \quad \mu_2=[-,-]:(V_0\oplus V_{-1})\otimes (V_0\oplus V_{-1})\rightarrow (V_0\oplus V_{-1})
\end{equation*}
such that the following \textit{Koszul conditions} are satisfied,
\begin{align*}
 &   [X,X']= - [X',X], \quad [X,Y] = - [Y,X], \quad \mu_1[X,Y]= [X,\mu_1Y], \quad [\mu_1 Y,Y']= [Y,\mu_1Y'],\\
   & [[X,X'],X'']+ [[X'',X],X']+ [[X',X''],X]=0, \qquad [[X,X'],Y]+ [[X,Y],X']+ [X,[X',Y]]=0,
\end{align*}
where $X,X',X''\in V_0$, $Y,Y'\in V_{-1}$.

It is convenient to write the graded bracket $\mu_2=[-,-]: V_i\otimes V_j \rightarrow V_{i+j}$, $-2\leq i+j \leq 0$, in terms of the degree $i,j \mod 2$ of $\g\cong V_{-1}\oplus V_0$, such that
\begin{equation}
    \mu_2(Y+X, Y’+X’)= [X,X']+ \big([X,Y']+ [Y,X']\big),\qquad X,X'\in V_0,~Y,Y'\in V_{-1}.
\end{equation} 
We shall also extend $\mu_1$ to the full space $V_0\oplus V_{-1}$ by $\mu_1(Y+X)=\mu_1Y$. 
\end{definition}

Given a Lie algebra crossed-module $\g=(\mathfrak{g}_{-1}\xrightarrow{t}\mathfrak{g}_0,\rhd,[-,-]_0)$, we simply identify $\g_{-1}=V_{-1},\g_0=V_0$ and $t=\mu_1$. Then, one reassembles the graded bracket $\mu_2$ from the bracket $[-,-]_0$ on $\g_0$ as well as the Lie algebra action $\rhd$ such that
\begin{equation*}
    \mu_2(Y+X, Y’+X’)= [X,X']_0+ \big(X\rhd Y' - X'\rhd Y\big),\qquad X,X'\in V_0,~Y,Y'\in V_{-1}.
\end{equation*}
It is then simple to check that the Lie algebra crossed-module conditions imply precisely the Koszul conditions; in particular, the Peiffer identity implies the Koszul identity
\begin{equation*}
   [\mu_1 Y,Y'] = [tY,Y'] = [Y,Y']\ovo = -[Y',Y]\ovo = -[tY',Y] = -[\mu_1 Y',Y] = [Y,\mu_1Y']
\end{equation*}
as required. Conversely, given a strict $L_2$-algebra, one may recover a Lie algebra crossed-module with the above procedure, provided one {\it defines} the bracket $[-,-]{\ovo}$ on $\g_{-1}$ by
\begin{equation}
    [Y,Y']{\ovo}\equiv [\mu_1Y,Y'],\label{koszulpeiffer}
\end{equation}
whence the Koszul conditions guarantee that this bracket is skew-symmetric and satisfies the Jacobi identity.

Due to this result, we will use "\textbf{Lie 2-algebras}" in the following to refer to both a Lie algebra crossed-module and a strict $L_2$-algebra. All Lie 2-algebras will be understood as \textit{finite-dimensional strict} $L_2$-algebras in this paper, unless otherwise specified.

\subsection{Lie 2-algebra cocycles} 
The following is a reformulation of the conditions "ID1-ID3" in \textbf{Theorem 2.15} of \cite{Chen:2012gz}. See also \cite{Bai_2013,Chen:2013}.
\begin{definition}
The following linear maps
\begin{equation}
    \delta_{-1}:\mathfrak{g}_{-1}\rightarrow \mathfrak{g}_{-1}^{\otimes 2},\qquad \delta_0:\mathfrak{g}_0\rightarrow (\mathfrak{g}_0\otimes\mathfrak{g}_{-1})\oplus(\mathfrak{g}_{-1}\otimes\mathfrak{g}_0) \nonumber
\end{equation}
on the Lie 2-algebra $\mathfrak{g}$ are called a {\bf Lie 2-algebra cocycle} iff the following conditions are satisfied
\begin{align}
 \delta_0 t &= (t\otimes 1 + 1\otimes t)\delta_{-1},\nonumber \\
 0&= (t\otimes 1 - 1\otimes t) \delta_0, \nonumber\\
 \delta_0([X,X']) &= (X\rhd \otimes 1 + 1\otimes \operatorname{ad}_X)\delta_0(X')  \nonumber \\
 &\qquad  - (X'\rhd \otimes 1 + 1\otimes\operatorname{ad}_{X'})\delta_0(X),\qquad \text{($(\operatorname{ad},\rhd)$-invariance)}\nonumber \\
    \delta_{-1}(X\rhd Y) &= (X\rhd \otimes 1 + 1\otimes X\rhd)\delta_{-1}(Y) \nonumber \\
    &\qquad  + \delta_0(X)(\rhd Y\otimes 1 + 1\otimes \rhd Y).\label{2cocycoh}
\end{align}
If $\delta$ satisfies furthermore the following {\it cobracket conditions}
\begin{eqnarray}
    0&=& \sum_{\text{cycl.}}((\delta_{-1}+\delta_0)\otimes 1)\circ \delta_0 = (\delta_{-1}\otimes1)\circ\delta_0 - (1\otimes\delta_0)\circ\delta_0 \nonumber\\
    &\qquad& - ~ (\tau\otimes 1)\circ(1\otimes \delta_0)\circ\delta_0 \nonumber\\
    0&=&\sum_{\text{cycl.}}(\delta_{-1}\otimes 1)\circ \delta_{-1} = (\delta_{-1}\otimes 1)\circ\delta_{-1} - (1\otimes\delta_{-1})\circ\delta_{-1} \nonumber\\
    &\qquad& -~ (\tau\otimes 1)\circ(1\otimes\delta_{-1})\circ\delta_{-1},\label{2cobracket}
\end{eqnarray}
where $\tau:\g\otimes\g\rightarrow\g\otimes\g$ swaps the tensor factors, then we call the data $(\g;\delta)$ a \textbf{Lie 2-bialgebra}.
\end{definition}

\paragraph{Lie 2-algebra structure on the shifted dual complex.}
Now consider the \textit{shifted} dual complex $\g^*[1] = \g^*_0\to \g^*_{-1}$. It comes naturally with a degree -1 canonical evaluation pairing form $\langle-,-\rangle:\g^*[1]\otimes \g\rightarrow k[1]$ given by $$\langle g+f,Y+X\rangle = f(Y)+g(X)$$ for each $g\in \g_0^*,f\in\g_{-1}^*,X\in\g_0,Y\in\g_{-1}$. The dual differential $t^T: \g_0^*\rightarrow \g_{-1}^*$ is determined through
\begin{eqnarray*}
    (t^Tg)(Y) = g(tY)\iff \langle t^T-,-\rangle = \langle -,t-\rangle.
\end{eqnarray*}
Note $\operatorname{deg}(\mathfrak{g}_0^*)=-1$ while $\operatorname{deg}(\mathfrak{g}_{-1}^*)=0$.  

It is easy to see that, by putting
\begin{equation*}
    \langle [f,f']^*,y\rangle=\langle f\otimes f',\delta_{-1}y\rangle,\qquad \langle f\otimes g,\delta_0x\rangle = \langle f\rhd^*g,x\rangle\,
\end{equation*}
and the linear dual $t^T:\mathfrak{g}_0^*\rightarrow\mathfrak{g}_{-1}^*$ of $t$, the dual complex $\g^*[1]$ becomes a Lie 2-algebra thanks to the conditions in \eqref{2cocycoh}.

The following notion is from \cite{Zucchini:2021bnn,chen:2022}.
\begin{definition}\label{balanced}
    We say a Lie 2-algebra $\g$ over a field $k$ of characteristic zero is \textbf{balanced} iff it is equipped with a non-degenerate invariant bilinear form of degree -1,
    $$\g\otimes\g\to k[1]\,.$$
\end{definition}
\noindent A balanced Lie 2-algebra $\g$ therefore comes with which linear isomorphisms $\g_0\cong\g_{-1}^*$ and $\g_{-1}\cong \g_0^*$, which is an important fact for later.

Derived (weak/quasi-weak) generalizations of Lie 2-(bi)algebras --- ie. 2-term $L_\infty$-(bi)algebras --- have also been studied in \cite{Bai_2013,Chen:2012gz}. In fact, more generally, a $n$-term $L_{\infty}$-algebra $L$ can be endowed with a cocycle which is dual to a $L_n$-algebra structure on its \textit{$n$-shifted} dual complex $L^*[n]$. Note that this notion of duality only exist for \textit{bounded} (both above and below) chain complexes.

\paragraph{The graded adjoint action.}

The 2-adjoint representation of $\mathfrak{g}$ on itself, denoted $$\operatorname{ad}:\g\to \operatorname{Der}\g,\qquad Y+X\mapsto ~_2\operatorname{ad}_{Y+X}=[Y+X,-]\,,$$ consists of the following components,
\begin{equation}\label{adjointaction}
    _2\operatorname{ad}=(\operatorname{ad}_0,\operatorname{ad}_{-1}):\mathfrak{g}\rightarrow \operatorname{End}\mathfrak{g},\qquad \begin{cases}\operatorname{ad}_0:\mathfrak{g}_0\rightarrow\operatorname{End}(\mathfrak{g}_0\oplus\mathfrak{g}_{-1}) \\ \operatorname{ad}_{-1}: \mathfrak{g}_{-1}\rightarrow\operatorname{Hom}(\mathfrak{g}_0,\mathfrak{g}_{-1})\end{cases},
\end{equation}
where
\begin{equation}
    \operatorname{ad}_0(X) = ([X,-], X\rhd -),\qquad \operatorname{ad}_{-1}(Y) = -\rhd Y\nonumber 
\end{equation}
for each $X\in\mathfrak{g}_0,Y\in\mathfrak{g}_{-1}$. They satisfy the following key identities
\begin{equation}
    \operatorname{ad}_Xt = t\chi_X,\qquad \operatorname{ad}_{-1}(Y)t = -\operatorname{ad}_Y,\qquad t\operatorname{ad}_{-1}(Y)= - \operatorname{ad}_{tY}\label{adrep}
\end{equation}
for each $X\in\mathfrak{g}_0,Y\in\mathfrak{g}_{-1}$, which come from the equivariance and the Peiffer identity \eqref{pfeif}. 

More details can be found in Example 2.3.7 of \cite{Angulo:2018}.

\paragraph{Poisson-Lie 2-groups.} As was proven in \cite{Chen:2012gz}, a cocycle $\delta$ of a finite-dimensional Lie 2-bialgebra $\g$ integrates to a multiplicative Poisson bivector field $\Pi$ on $G$ given by
\begin{equation*}
    \Pi_{(y,x)} =(L_*)_{y,x)}\hat\delta_{(y,x)},
\end{equation*}
where $\hat\delta$ is the Lie 2-group cocycle corresponding to $\delta$, and $_2\operatorname{Ad}: G\to \operatorname{Aut}\g$ is the Lie 2-group adjoint representation integrating \eqref{adrep} (see the paragraph surrounding \textbf{Definition \ref{def:invariant}}). This makes $(G,\Pi)$ into a Poisson-Lie 2-group. 

\paragraph{The quadratic 2-Casimir.} The following notion can be found in \cite{chen:2022}.
\begin{definition}\label{2casimir}
    A \textbf{quadratic 2-Casimir} $K\in \g\odot\g$ of a Lie 2-algebra $\g$ is a graded-invariant (under \eqref{adjointaction}) graded-symmetric non-degenerate quadratic form on $\g$ of degree 1, which is symmetric under $t$,
    $$(t\otimes 1)K = (1\otimes t)K\,.$$
\end{definition}
\noindent Given a \textit{quadratic 2-Casimir} $K$, one can induce a non-degenerate invariant bilinear pairing 
$\langle-,-\rangle_K:\g\otimes \g\rightarrow k[1]$ of degree -1 for which 
\begin{equation}
    \langle tY,Y'\rangle  = \langle Y,tY'\rangle,\qquad \forall~ Y,Y'\in\g_{-1},\label{tsymm}
\end{equation}
making $\g$ balanced as in \textbf{Definition \ref{balanced}}. This identifies $\g$ with its own dual Lie 2-algebra  $\mathfrak{g}^*[1]$, where the explicit identification is given by $$X\mapsto f_X =\langle -,X\rangle\in\g_{-1}^*,\qquad Y\mapsto g_Y = \langle Y,-\rangle\in\g_0^*,$$ such that $K$ is related to the evaluation pairing via
\begin{equation*}
    (g_Y+f_X)(X'+Y') = f_X(X') + g_Y(Y')= \langle Y+X,Y'+X'\rangle_K,
\end{equation*}
as usual. We will explain in \S \ref{2grclrmat} how quadratic 2-Casimirs are related to solutions of the {2-graded classical Yang-Baxter equations} (\textbf{Definition \ref{2rmatrixdef}}).

\subsection{Differential-graded Poisson structures}\label{2grpois2alg}
We start by recalling the notion of dg shifted Poisson algebras as treated in \cite{pantev2013shifted,Voronov:2008,pridham2018outline,Seol2022-sr,Calaque2015ShiftedPS,SAFRONOV2021107633}. We note that the notion of dg/derived manifolds have a long history since the work of Kontsevich-Soibelman (eg. \cite{Kontsevich:2006jb} and the book draft \cite{Kontsevich_deformation}), but we shall follow more recent treatments.

By a \textbf{differential-graded (dg) manifold} $M$ we will mean a real smooth manifold together with a graded commutative $\mathbb{R}$-algebra $\mathcal{O}_M = \mathcal{O}_M^0 \to \mathcal{O}_M^1\to\dots$ of sheaves $\mathcal{O}_M^j$ of dg commutative (locally-semifree) algebras on $M$. 

The usual smooth functions $\mathcal{O}_M^0$ on $M$ sit at degree-0, and the cochain differential  ${\bf t}: \mathcal{O}^\bullet_M\to \mathcal{O}^{\bullet+1}_M$ is a smooth derivation satisfying the Leibniz rule
\begin{equation*}
    {\bf t}(FG) = ({\bf t} F)G + (-1)^{\text{deg}F}F({\bf t} G)
\end{equation*}
on homogeneous elements $F,G\in \mathcal{O}_M$. Loosely speaking, such structures are also sometimes called \textbf{$\mathbb{NQ}$-manifolds} \cite{Severa:2005}.

We denote by $C(M)^\bullet=\Gamma(M,\mathcal{O}_M^\bullet)$ (or simply $C(M)$) the differential-graded commutative algebra (dgca) of global sections of $\mathcal{O}_X$.
\begin{definition}
    Let $M$ denote a dg manifold. The \textbf{smooth 1-forms $\Omega^1_M$} on $M$ are defined as the $C(M)$-module in (super) chain complexes generated by elements $da$ for $a\in C(M)$ (with the same degree/parity), such that
    \begin{enumerate}
        \item $d(ab) = (da)b + (-1)^{\bar a}a(db)$ where $\bar a$ denotes the parity of $a\in C(M)$,
        \item for $(a_1,\dots,a_n)\in C(M)^0$ in degree 0 (analogues of the coordinate functions on $M$), we have
        \begin{equation*}
            df(a_1,\dots,a_n) = \sum_{i=1}^n \frac{\partial f}{\partial x_i}(a_1,\dots,a_n) da_i\,
        \end{equation*}
        for all $f\in C^\infty(M)$,
        \item $d$ commutes with the cochain differential ${\bf t}$, $d({\bf t}a)= {\bf t}(da)$.
    \end{enumerate}
    The \textbf{smooth 1-vectors} $T_M$ on  $M$ is the dg $C(M)$-module internal-hom $\underline{\operatorname{Hom}}_{C(M)}(\Omega^1_M,C(M))$.
\end{definition}

As an important example, for a finite-dimensional real Lie algebra $\g$ we can consider a dg manifold $M=B\g$ whose underlying manifold is a point, with the sheaf $\mathcal{O}_M$ in cochain complexes given by the Chevalley-Eilenberg complex
\begin{equation*}
    \operatorname{CE}^\bullet(\g) = \R\to \g^*\to \wedge^2 \g^*\to \wedge^3\g^*\,\to\cdots.
\end{equation*}
Similarly, if $\g = \g_{-1}\xrightarrow{t}\g_0 $ were a Lie 2-algebra, then we have a similar construction for the dg manifold $B\g$ with the sheaf $\mathcal{O}_M = \operatorname{CE}^\bullet(\g)$.\footnote{This is in fact one of the most classic examples of a dg manifold, appearing in many classic references such as \cite{Severa:2005,Alexandrov:1995kv,Kontsevich:2006jb,Kontsevich_deformation}.} Each piece of the exterior product $\wedge^n \g^*[1]$ is a cochain complex, which we must totalize in $\operatorname{CE}^\bullet(\g)$. For instance, we have (recall $\g^*[1] =\g^*_0\xrightarrow{t^T} \g_{-1}^*$ is the \textit{shifted} dual complex)
\begin{equation*}
    \g^*[1]\wedge \g^*[1]= \g_{0}^*\odot \g_{0}^*\xrightarrow{D_{t^T}} (\g_0^*\otimes\g_{-1}^*)\oplus (\g_{-1}^*\otimes\g_0^*)\xrightarrow{D_{t^T}} \g_0^*\wedge\g_0^*
\end{equation*}
is a 3-term chain complex (cf. \cite{Bai_2013}), where  $D_{t^T} = 1\otimes t^T\pm t^T\otimes 1 $ is the cochain differential with the sign $\pm$ depends on the cohomology degree, and $\odot$ denotes the \textit{symmetric} tensor product.

Indeed, it is crucial to note that the exterior product of components of odd cohomologial degree is symmetric, hence we can package $\mathcal{O}_{B\g}$ in the following compact way,
\begin{eqnarray}
    \mathcal{O}_{B\g}^\bullet \cong \operatorname{Sym}^\bullet(\g_{0}^*) \otimes \bigwedge^\bullet \g_{-1}^*\,.\label{paritydecomp}
\end{eqnarray}
By extending the cohomological differential $t: \g_{-1}\to \g_0$ to act trivially on $\g_0$, we equip $\mathcal{O}_{B\g}$ with differential ${\bf t} = Q - t: \mathcal{O}_{B\g}^\bullet\to \mathcal{O}_{B\g}^{\bullet+1}$ given by $t$ as well as the graded Chevalley-Eilenberg differential $Q$ given by the graded cobracket $\delta$.

\medskip

Recall the following notions from \S 2.1 of \cite{pridham2018outline}.
\begin{definition}
    Let $(M,\mathcal{O}_M)$ denote a dg manifold. The \textbf{$n$-shifted polyvector fields} on $M$ is the cochain complex\footnote{Recall that the notation "$[n]$" means to shift \textit{down} the degree by $n$, and "$[-n]$" shifts the degree up.}
    \begin{equation*}
        \mathfrak{X}_M(n) = \prod_j \operatorname{Sym}^j_{C(M)}(T_M[-n-1])
    \end{equation*}
    of $C(M)$-modules obtained from graded symmetric products of $T_M$. In particular, the \textbf{$n$-shifted $k$-vector fields} on $M$ are defined as the associated graded  $\mathfrak{X}_M^{\geq k}(n)/\mathfrak{X}_M^{\geq k+1}(n)$ of the filtration
    \begin{equation*}
        \mathfrak{X}_M^{\geq k}(n) = \prod_{j\geq k}\operatorname{Sym}^j_{C(M)}(T_M[-n-1])\,.
    \end{equation*}
\end{definition}
\noindent The Lie bracket on $T_M$ extends to a differential-graded bracket (the Schouten-Nijenhuis bracket)  $$\mathfrak{X}_M(n)\times \mathfrak{X}_M(n)\to \mathfrak{X}_M(n)[-n-1]\,.$$ Thus, shifting the degrees back up, $\mathcal{X}_M(n)[n+1]$ is a (super) differential-graded Lie algebra (dgla) together with a compatible $C(M)$-module structure.

For any dgla $L$, we denote by $\operatorname{MC}(L)\subset L_1$ the space of \textit{Maurer-Cartan elements} satisfying
\begin{equation*}
    d\omega+[\omega,\omega] = 0,
\end{equation*}
where $[-,-]$ is the graded Lie bracket.
\begin{definition}
    A \textbf{$n$-shifted Poisson bivector field} $\Pi\in \mathfrak{X}^2_M(n)$ on $M$ is a Maurer-Cartan element in $\mathfrak{X}^2_M(n)[-n-1]$.
\end{definition}

\medskip

Let us unpack these definitions in the example of interest, following \textit{Examples 2.9, 2.10} in \cite{pridham2018outline} and \S 2 of \cite{SAFRONOV2021107633}. First for an ordinary Lie algebra $\g$, taking $(M,\mathcal{O}_M) = (B\g,\operatorname{CE}(\g))$ and $n=1$, we have that a 1-shifted Poisson bivector on $B\g$ is defined by the following ingredients,
\begin{equation*}
    \delta\in\g^*\otimes \wedge^2\g,\qquad \phi\in \wedge^3\g
\end{equation*}
which satisfies 
\begin{equation*}
    Q\delta=0,\qquad \frac{1}{2}[\delta,\delta]= Q\phi
\end{equation*}
under the Chevalley-Eilenberg differential $Q$. By viewing $\delta: \g\to \wedge^2\g$ as a cochain, these are nothing but very compact notations for the $\operatorname{ad}$-invariance of $\delta$, as well as the failure of the cobracket/coJacobi identities of $\delta$ measured by the trivector $\phi$. In other words, the 1-shifted Poisson dg manifold $(B\g,\Pi)$ is nothing but a \textit{quasi}-Lie bialgebra.

Similarly, for a Lie 2-algebra $\g$, a 1-shifted Poisson bivector $\Pi\in\mathfrak{X}^2_{B\g}(n)[-n-1]$ is characterized by
\begin{enumerate}
    \item an element $\delta\in \g^*[1]\otimes \g$ which can be viewed as a linear map $\delta: \g[-1]\to \g\wedge\g$,
    \begin{equation*}
        \delta_{-1}: \g_{-1}\to \g_{-1}\otimes \g_{-1},\qquad  \delta_0: \g_0\to (\g_{-1}\otimes\g_0)\oplus(\g_0\otimes\g_{-1})\,,
    \end{equation*}
    \item a trivector $\phi\in \wedge^3\g$ which we for simplicity restrict to $\phi\in\operatorname{Sym}^3(\g_{-1})$,
\end{enumerate}
for which the condition $Q\delta=0$ is nothing but the equations \eqref{2cocycoh}, and the condition $\frac{1}{2}[\delta,\delta]=Q\phi$ is nothing but the \textit{weak} version of the 2-cobracket condition \eqref{2cobracket} --- that is, we have
\begin{eqnarray}
    D_t\phi &=& \sum_{\text{cycl.}}((\delta_{-1}+\delta_0)\otimes 1)\circ \delta_0 = (\delta_{-1}\otimes1)\circ\delta_0 - (1\otimes\delta_0)\circ\delta_0 \nonumber\\
    &\qquad& - ~ (\tau\otimes 1)\circ(1\otimes \delta_0)\circ\delta_0 \nonumber\\
    \phi &=&\sum_{\text{cycl.}}(\delta_{-1}\otimes 1)\circ \delta_{-1} = (\delta_{-1}\otimes 1)\circ\delta_{-1} - (1\otimes\delta_{-1})\circ\delta_{-1} \nonumber\\
    &\qquad& -~ (\tau\otimes 1)\circ(1\otimes\delta_{-1})\circ\delta_{-1}\,.\label{weakcobracket}
\end{eqnarray}
In other words, $(\g,\delta,\phi)$ is a \textit{quasi}-Lie 2-bialgebra \cite{Chen:2012gz,Chen:2013}.

By the same token, a {Poisson-Lie 2-group} $\mathbb{G}$ is precisely a 1-shifted Poisson structure on the Lie 2-group $\mathbb{G}$, considered as a dg manifold equipped with the sheaf $\mathcal{O}_\mathbb{G}=\overleftarrow{\mathfrak{X}}_\mathbb{G}$ of left-invariant vector fields \cite{Chen:2012gz} (see also \S 3 of \cite{SAFRONOV2021107633}). 

\medskip

Now recall that, given a (strict, ie. $\phi=0$) Lie 2-bialgebra $(\g,\delta)$, its shifted dual complex $\g^*[1]$ itself is a Lie 2-(bi)algebra. We can therefore consider the 1-shifted Poisson structure on $B(\g^*[1])$ induced from a so-called \textit{classical 2-$r$-matrix} \cite{Bai_2013}.

\medskip

We make a brief comment here that the above theory of dg manifolds extends naturally to the broader context of \textit{derived} manifolds and derived Poisson stacks. However, the results of the following sections rely heavily on a systematic theory of 2-graded classical Yang-Baxter equations, which is not yet available in the weak case (for some recent developments in this direction, see \cite{Cirio:2012be,Kemp_2025}).

\subsection{The 2-graded classical \texorpdfstring{$r$}{Lg}-matrix}\label{2grclrmat}
The theory of 2-bialgebras and the associated 2-graded classical $r$-matrix $R\in \mathfrak{g}_0\otimes\mathfrak{g}_{-1}$ have been studied previously in \cite{Bai_2013,Chen:2012gz,chen:2022}. 
\begin{definition}\label{2rmatrixdef}
  Denote by
 $$D_t^\pm= t\otimes 1\pm1\otimes t$$ 
 the differential in the tensor product complex $\g^{2\otimes}$.  A \textbf{2-graded classical $r$-matrix} $R\in \g_{-1}\otimes\g_0\oplus\g_0\otimes\g_{-1}$ on a finite-dimensional Lie 2-algebra $\g$ is an element of degree -1 in $\g^{\otimes 2}$ which satisfies  the {\bf 2-graded classical Yang-Baxter equations} (2-CYBE)
\begin{gather}
    \llbracket R,R\rrbracket = [R_{12},R_{13}] +[R_{12},R_{23}]  + [R_{13},R_{23}] = 0,\nonumber\\
    D_t^-R= 0,\label{2cybe}
\end{gather}
where $\llbracket-,-\rrbracket$ is the graded Schouten bracket (by treating $\g$ as left-invariant vector fields on its Lie 2-group $G$ as a groupoid) \cite{Bai_2013,Chen:2012gz}. 
\end{definition}
\noindent Note the second condition $D_t^-R=0$ has no 1-graded analogue, and is necessary for our theory (see for instance \textbf{Proposition \ref{dgmap}}).

Given a 2-graded $r$-matrix $R$, write $R = R_1 \oplus R_2$, where $R_1\in \g_{-1}\otimes\g_0$ and $R_2\in\g_0\otimes \g_{-1}$, then we form the skew-symmetric and symmetric pieces
\begin{equation*}
    R^\wedge = (R_1 - R_2') \oplus (R_2 - R_1'),\qquad R^\odot = (R_1+R_2')\oplus (R_2+R_1') \label{skew/sym}
\end{equation*}
of $R$, where $R'$ denotes the tensor with swapped tensor factors.\footnote{That is, we have $R_1'\in\g_0\otimes\g_{-1}$ and $R_0'\in\g_{-1}\otimes\g_0$.} It then follows clearly that we have $R = R^\wedge + R^\oplus$, and \eqref{2cybe} becomes $$\llbracket R^\wedge ,R^\wedge\rrbracket=-\llbracket R^\odot,R^\odot\rrbracket\,.$$

Now It was shown in \cite{Bai_2013} that, under certain mild technical conditions, there is a decomposition $R = r-D_t^+ \r$ where
\begin{equation}
    r\in (\mathfrak{g}_0\otimes\mathfrak{g}_{-1})\oplus(\mathfrak{g}_{-1}\otimes\mathfrak{g}_0),\qquad \r\in\mathfrak{g}_{-1}^{\otimes 2}.\nonumber
\end{equation}
For later applications, it would be convenient to suppose that the skew-symmetric and symmetric components of $R$ also decompose accordingly \cite{chen:2022},
\begin{equation}
    R^\wedge = r^\wedge - D_t^+\r^\wedge ,\qquad R^\odot = r^\odot - D_t^+\r^\odot,\label{2Rmatdecomp} 
\end{equation}
where $r^\wedge,r^\odot$ are defined as in \eqref{skew/sym}, and similarly for $\r^\wedge,\r^\odot$. Notice under  \eqref{2Rmatdecomp}, the condition $D_t^-R=0$ only concerns the factor $r$, as $D_t^-D_t^+=0$ \cite{Bai_2013}.


\medskip

These elements $R^\odot,R^\wedge$ shall play a central role in the following. 

\paragraph{Skew-symmetric component $R^\wedge$.} 
It can be shown that the skew-symmetric piece $R^\wedge$ yields a 2-coboundary $\delta=(\delta_{-1},\delta_0)=dR^\wedge$ on $\mathfrak{g}$ satisfying \eqref{2cocycoh}, \eqref{2cobracket}, given explicitly by\footnote{Note here that $[-,-]$ is the graded Lie bracket on $\mathfrak{g}$, which includes $[-,-]|_{\mathfrak{g}_0}$ as well as the crossed-module action $\rhd$.} \cite{Bai_2013} 
\begin{equation}
\delta_0(X) = [X\otimes 1+1\otimes X,R^\wedge],\qquad \delta_{-1}(Y) = [Y\otimes 1+1\otimes Y,R^\wedge].\label{2cocy}
\end{equation}
Hence, the data $(\g,\delta=dR^\wedge)$ defines a Lie 2-bialgebra. The corresponding Lie 2-group $G$ would then also have a Poisson-Lie 2-group structure.

\medskip

What is important for this paper is the fact that the 2-$R$-matrix $R$ determines an \textit{alternative} dg $R$-bracket $[-,-]_R: \g\otimes\g\to\g$ on $\g$, through which the Kostant-Kirilov symplectic structure can be generalized. We shall do this in \S \ref{2kkbracket}.

\paragraph{Symmetric component $R^\odot$.} Recall from the Lie 1-algebra case that the classical Yang-Baxter equation constrains the symmetric piece of the classical $r$-matrix to be $\operatorname{ad}$-invariant. Similarly, \eqref{2cybe} constrains $R^\odot$ to be invariant under the 2-adjoint representation $_2\operatorname{ad}$ \eqref{adrep}. Moreover, $R^\odot$ must also satisfy the equation $D_t^-R^\odot=0$.

Now provided $R^\odot\in\mathfrak{g}_0\odot\mathfrak{g}_{-1}$ is non-degenerate, it determines a non-degenerate bilinear form $\langle-,-\rangle=\langle-,-\rangle_{\odot}$ on $\g$, up to a scalar factor. Similar to the Lie 1-algebra case, this scalar factor can be fixed by choosing $R^\odot$ to be a {quadratic 2-Casimir} of $\mathfrak{g}$  as in \textbf{Definition \ref{2casimir}}. In other words,  a solution $R$ of the 2-graded classical Yang-Baxter equations in fact makes $\g$ into a \textit{balanced} Lie 2-bialgebra.

\begin{remark}\label{zamTE}
    It is important to note that the 2-graded $r$-matrix does \textit{not} give rise to solutions of the {Zamolodchikov tetrahedron equations} \cite{Zamolodchikov:1980,Kuniba_2023}, but it is known that they both are crucial data involved in the structure of braided monoidal 2-categories (cf. \cite{Kapranov:1994,neuchl1997representation,Chen1:2025?}) --- the 2-$r$-matrix is an "infinitesimal" incarnation of the braiding \cite{Kemp_2025}. The relation between the 2-graded Yang-Baxter equations and the tetrahedron equations is not yet clear at the moment.
\end{remark}

\medskip

The interplay of the quadratic 2-Casimir $R^\odot$ and the 2-coboundary $R^\wedge$, as dictated by the 2-graded classical Yang-Baxter equations \eqref{2cybe}, is the crucial datum in the following that allows us to develop our theory of higher integrability.

\section{2-graded integrability}\label{2integra}
The goal in this section is to lay out the general theory of 2-graded integrable systems, and derive an appropriate notion of a "2-graded Lax equation" as a categorification of the usual Lax equation. Then, once this is achieved, we specialize to the dual crossed-module $\mathfrak{g}^*[1]$ and construct a 2-graded Lax pair on it, in analogy with the 1-algebra case introduced in \S \ref{poisbialg}. We then work to prove that it does in fact satisfy the 2-graded Lax equations.

Consider a dg manifold $M$ equipped with a 1-shifted Poisson bivector $\Pi\in\mathfrak{X}_M^2(1)[-2]$ as in \S \ref{2grpois2alg}. We define a curve $\gamma$ on the dg manifold as a curve $\gamma:[0,1]\to M$ on the underlying manifold $M$, together with an induced pullback map $\gamma^*:C(M)\to\Omega^\bullet([0,1])$ on sheaves of dgca's, where $\Omega^\bullet([0,1])$ is the de Rham differential forms on $[0,1]$. In other words, we treat  $[0,1]$ as a dg manifold with its shifted tangent complex $T_{[0,1]}[-1]$.

\begin{definition}\label{dgcurve}
    Let $(M,\mathcal{O}_M)$ denote a dg Poisson manifold. For a Hamiltonian $H\in C(M)$, we say the integral curve $\gamma:[0,1]\to M$ is \textbf{generated by $H$} iff it is generated by the dg Hamiltonian vector field $X_H=\{H,-\}=\Pi(H\otimes-)$. 
\end{definition}
\noindent We keep in mind that the vector field $X_H$ has the opposite grading parity compared to $H$; see \S 4 of \cite{Voronov:2008}.

\subsection{2-Lax pair}\label{2poiss}
Consider a dg space $M$ and its smooth global sections $C(M) = \Gamma(M,\mathcal{O}_M)$. For $M$ equipped with a dg Poisson structure, we let $(C(M),\{-,-\})$ denote the dg Poisson algebra arising from it.

Consider Lie 2-algebra $\mathfrak{g}$-valued functions on $M$. In the dg context, they are elements in the tensor product $C(M)\otimes\g$, which upon totalization is the cochain complex
\begin{gather}
    \cdots\to\LaTeXunderbrace{(C^{-1}(M)\otimes \mathfrak{g}_{-1})\oplus (C^{-2}(M)\otimes\g_0)}_{\text{deg-}(-2)} \xrightarrow{D^+}  \LaTeXunderbrace{(C^0(M)\otimes\mathfrak{g}_{-1})\oplus(C^{-1}(M)\otimes\mathfrak{g}_0)}_{\text{deg-}(-1)} \nonumber \\
    \xrightarrow{D^-} \LaTeXunderbrace{C^0(M) \otimes\mathfrak{g}_0}_{\text{deg-}0}\to 0
    \label{tensorxmod}
\end{gather}
with the chain maps $D^\pm = 1\otimes t\pm  {\bf t}\otimes 1$ (with sign depending on the degree), where $\textbf{t}$ is the differential on $C(M)$. The graded Lie bracket $[-,-]$ on $\g$, together with the graded Poisson bracket $\{-,-\}$ on $C(M)$, endow this complex with two individual Lie 2-algebra structures.

Let $H\in C(M)$ be a smooth function on $M$, which we call the \textit{Hamiltonian}.
\begin{definition}
We say a tuple of elements $L,P\in C(M)\otimes\g$ is a \textbf{2-Lax pair} of the Hamiltonian system $(M,\{-,-\},H)$ iff, on an integral curve $\gamma:[0,1]\to M$ generated by $H$ (as in \textbf{Definition \ref{dgcurve}}), they satisfy the {\bf dg-Lax equation}
\begin{equation}
    \partial_t (\gamma^*L) \equiv \gamma^*(\{H,L\}) = [\gamma^*P,\gamma^*L],\label{2grlax}
\end{equation}
where $\{-,-\},[-,-]$ are the graded Poisson/Lie brackets on the complex \eqref{tensorxmod}.
\end{definition}
\noindent Due to the graded nature of the Hamiltonian $H$, the dynamics along the integral curve $\gamma$ it generates is also graded, in the sense that there are essentially a chain complex worth of "energies" under a single time parameter.

\medskip

We shall see that the canonical 2-Lax pair $L,P$ constructed in \S \ref{2laxconstruction} makes crucial use of the fact that they form homogeneous elements of the cochain complex $C(M)\otimes\g$.

\subsection{Conserved quantities}
Recall that in the 1-algebra case, the trace polynomials $f_k$ of the Lax function $L$ are constants of motion. We wish now to investigate the analogous notion of "2-graded integrability" afforded by the 2-Lax equations  \eqref{2grlax}. Toward this, we must first explain how to construct trace polynomials in the 2-graded context and hence the relevant concept of 2-representation in our context.

\paragraph{2-representations.} 
Let $V=V_{-1}\xrightarrow{\partial} V_0$ denote a 2-term complex of vector spaces. The following is from \cite{Angulo:2018} (see also \cite{Baez:2003fs,heredia2016representations2groupsbaezcrans2vector}).
\begin{definition}
The space of endomorphisms $\mathfrak{gl}(V):\operatorname{End}^{-1}(V)\xrightarrow{\delta}\operatorname{End}^0(V)$ of $V$ is a 2-graded space
\begin{equation}\label{eq:MN}
    \operatorname{End}^{-1}(V) = \operatorname{Hom}(V_0,V_{-1}),\qquad \operatorname{End}^0(V) =\{M+N\in \operatorname{End}(V_{-1})\oplus \operatorname{End}(V_0)\mid N\partial = \partial M\}, 
\end{equation}
equipped with the following (strict) Lie 2-algebra structure \cite{Bai_2013,Angulo:2018}
\begin{eqnarray}
\delta:\operatorname{End}^{-1}(V)\rightarrow \operatorname{End}^0(V),&\quad& \delta(A) = A\partial + \partial A,\nonumber\\
{[M+N,M'+N']}_C = [M,M']+[N,N'],&\quad& (M+N)\rhd_C A = MA - AN,\nonumber\\
{[A,A']}_C = A\partial A' - A'\partial A,\nonumber
\end{eqnarray}
for each $M+N\in\operatorname{End}^0(V),~A\in\operatorname{End}^{-1}(V)$.
\end{definition}

Generally, $\operatorname{End}(V)$ of a dg vector space $V$ will have the structure of a dgla which we could use. We truncate $V$ to only have two terms for simplicity.

\begin{definition}\label{2repdef}
A (strict) {\bf 2-representation} $\rho:\g\rightarrow\mathfrak{gl}(V)$ is a Lie 2-algebra homomorphism such that the following square
\begin{equation}
\begin{tikzcd}
    \g_{-1}\arrow[r,"t"] \arrow[d,"\rho_1"] & \g_0 \arrow[d,"\rho_0"] \\ 
    \operatorname{End}^{-1}(V) \arrow[r,"\delta"] & \operatorname{End}^0(V)  
\end{tikzcd}    \label{commsq}
\end{equation}
commutes. More explicitly, we have $\rho=(\rho_0, \rho_1)$ with $\rho_0(X)=(\rho_0^0(X), \rho_0^1(X))\in \operatorname{End}^0(V)$ and $\rho_1(Y)\in \operatorname{End}^{-1}(V)$ for each $X\in\g_0,Y\in\g_{-1}$, such that the following conditions
\begin{align}
&    \rho_0^0(tY) = \partial\rho_{1}(Y),\qquad \rho_0^{1}(tY) = \rho_{1}(Y)\partial,\nonumber\\
 &   \rho_{1}(X\rhd Y) = (\rho_0X)\rhd_C\rho_{1}Y= \rho_0^1(X)\rho_{1}(Y) - \rho_{1}(Y)\rho_0^0(X)\label{axioms2rep}
\end{align}
are satisfied.
\end{definition}
\noindent Furthermore, $\rho_0=\rho_0^{1}+\rho_0^0$ represents $\g_0$ on respectively $V_{-1}$ and $V_0$, with $\partial$ as the intertwiner. Elementary examples of 2-representations include the adjoint representation of $\g$ on itself, or the coadjoint representation of $\g$ on the dual Lie 2-algebra $\g^*[1]$; see \cite{Bai_2013,chen:2022} for details of these examples.

\smallskip

Any 2-representation $\rho$ as defined above gives rise to a "genuine" matrix representation $\rho^\text{gen}$ (terminology from \cite{Angulo:2018}) on the direct sum $V_{-1}\oplus V_0$, which takes the form of a block matrix
\begin{equation}
    \rho^\text{gen}(L)= \begin{pmatrix}\rho_0^{1}(L_0+tL_{-1}) & \rho_{1}(L_{-1}) \\ 0 & \rho_0^0(L_0)\end{pmatrix}\in C(M)\otimes\mathfrak{gl}(V_{-1}\oplus V_0), \quad L_0\in\g_0, \quad L_{-1}\in\g_{-1}\label{matrep},
\end{equation}
where $L_{-1},L_0$ denote the graded components of $L$ that take values in $\g_{-1},\g_0\subset \g$, respectively. It can be shown that we have $\rho^\text{gen}([L,P]) = [\rho^\text{gen}(L),\rho^\text{gen}(P)]_C$, where $[-,-]_C$ is the matrix commutator bracket on $\mathfrak{gl}(V_{-1}\oplus V_0)$.

\paragraph{Invariance under 2-adjoint action.}
Clearly, the 2-adjoint action $_2\operatorname{ad}$ \eqref{adrep} is a Lie 2-algebra 2-representation $\g\to \mathfrak{gl}(\g)$. By the Lie theorem for Lie 2-groups (assuming the complex underlying $\g$ is finite-dimensional), it integrates to the adjoint Lie 2-group 2-representation $_2\operatorname{Ad}: G\to \operatorname{GL}(\g)$, which can be defined strictly through the perspective of crossed-modules (see Example 3.3.9 in \cite{Angulo:2018} and \S 1.7 of \cite{Wagemann+2021}). 

It can be extended to a graded automorphism on the Chevalley-Eilenberg complex $B(\g^*[1])=\operatorname{CE}(\g^*[1])$ through the diagonal.
\begin{definition}\label{def:invariant}
 A function $H\in C(B(\g^*[1]))=\operatorname{CE}(\g^*[1])$ is invariant iff 
\begin{equation}
    H\in \operatorname{CE}(\g^*[1])^{G}\label{adinv}
\end{equation}
lies in the strict fixed-points under the 2-adjoint action $_2\operatorname{Ad}: G\to \operatorname{GL}(\g)$. 
\end{definition}
Note that, in the derived context, the adjoint action by weak Lie 2-algebras (or more generally $L_\infty$-algebras)/weak Lie 2-groups is defined only up to homotopy; the case incorporating the Jacobiator/homotopy map was studied in \S 4 of \cite{Chen:2023tjf}. See also \textit{Remark \ref{homotopyinvariants}}.


\paragraph{Constants of (graded) motion.} We are now ready to characterize the notion of conserved quantities inherited from the construction of 2-representation built out on  2-vector spaces of the Baez-Crans type. 

For a strict Lie 2-algebra $\g$, there is a notion of a function $f:\g \to k$ being \textit{strict} invariant,
$$f([X,Y]_C) = 0,\qquad \forall~ X,Y\in\g.$$
We call such invariant functions the\textbf{class functions} on $\g$.\footnote{For $\g=\mathfrak{gl}(V)$ of a 2-term cochain complex $V$ and $f$ linear, in particular, this strict invariance is the same as cyclicity
\begin{equation*}
    f([M,N]_C) = f(MN-NM)= f(MN)-f(NM)=0,
\end{equation*}
where $M,N\in\mathfrak{gl}(V)$ and $[-,-]_C$ is the commutator.} 
\begin{theorem}
Let $V$ be a 2-representation of a Lie 2-algebra $\g$ and $\chi_V:\mathfrak{gl}(V)\rightarrow k$ denote a class function on $\mathfrak{gl}(V)$. Given a 2-Lax pair $L,P\in C(M)\otimes \g$, the 2-Lax equation  \eqref{2grlax} implies that the polynomials
\begin{equation}
    \mathcal{F}_m = \chi_V(\rho(L)^m)\label{tracepoly}
\end{equation}
are constant along any integral curve $\gamma$ generated by the Hamiltonian $H$ (as in \textbf{Definition \ref{dgcurve}}) on $M$.
\end{theorem}
\begin{proof}
The proof runs in exact analogy with the 1-algebra case \cite{Meusburger:2021cxe}. Take an integral curve $\gamma:[0,1]\to M$ generated by $H$ and its induced dgca map $\gamma^*: C(M)\to \Omega^\bullet([0,1])$. It commutes with $\rho$, since they act on different tensor factors in $C(M)\otimes \g $. Then, from \eqref{2lax} and the invariance of $\chi_V$, we have
\begin{align*}
    \partial_t{(\gamma^*\mathcal{F}}_m) &= \partial_t \chi_V (\rho[\gamma^*L,\gamma^*P]^m) =  \chi_V (\partial_t[\gamma^*\rho(L),\gamma^*\rho(P)]_C^m)\\
    &=\sum_{i=0}^{m-1}\chi_V\big(\rho(\gamma^*L)^i\rho(\partial_t(\gamma^* L))\rho(\gamma^*L)^{m-i-1}\big) = m\, \chi_V\big(\rho(\gamma^*L)^{m-1}\rho([\gamma^*L,\gamma^*P])\big).
\end{align*}
By the fact that $\rho$ is a homomorphism of Lie 2-algebras, we have $\rho([L,P]) = [\rho(L),\rho(P)]_C$ and hence \begin{eqnarray*}
    \chi_V\big(\rho(\gamma^*L)^{m-1}\rho([\gamma^*L,\gamma^*P])\big) &=& \chi_V(\rho(\gamma^*L)^{m-1}[\rho(\gamma^*L),\rho(\gamma^*P)]_C) \nonumber\\
    &=&\chi_V([\rho(\gamma^*L)^m,\rho(\gamma^*P)]_C) = 0,
\end{eqnarray*}
again from the invariance of $\chi_V$.
\end{proof}

Note that the conservation of these trace polynomials is independent of the choice of the 2-representation $\rho$. However, what exactly is being conserved does depend on $\rho$. The conservation of these eigenvalues can be understood as the notion of "2-graded integrability" that the 2-Lax pair in {\bf Definition \ref{2grlax}} affords. 

\medskip

To see the conserved quantities explicitly for a given 2-representation $\rho$, we make use of the genuine matrix representation $\rho^\text{gen}$ \eqref{matrep} corresponding to $\rho$. It is obvious that the usual matrix trace form $\chi_V=\operatorname{tr}_{V_{-1}\oplus V_0}$ on $\mathfrak{gl}(V_{-1}\oplus V_0)$ is a class function \cite{Angulo:2018}. As such, the 2-graded trace polynomials 
\begin{eqnarray*}
    \mathcal{F}_m = \operatorname{tr}_{V_{-1}\oplus V_0}(\rho^\text{gen}(L)^m)
\end{eqnarray*}
of a Lax operator $L\in C(M)\otimes\g$ are conserved for any $m\in\mathbb{Z}_{\geq 0}$. Since $\rho_\text{gen}(\gamma^*L)$ is block triangular, so are its  powers. A fundamental result in linear algebra then tells us that the conserved quantities $\mathcal{F}_m$ can be organized in terms of the combined eigenvalues of the diagonal blocks of $\rho^\text{gen}(L)$:
\begin{equation}
    \operatorname{Eigen}\rho^\text{gen}(L)^m = \operatorname{Eigen}\rho_0^{1}(L_0+tL_{-1}) \coprod \operatorname{Eigen}\rho_0^0(L_0).\nonumber
\end{equation}
These are example of the conserved quantities associated to the 2-Lax equation \eqref{2grlax} that one can always compute, using the genuine representation \eqref{matrep}.

\medskip

Let us consider some special cases. When $t=0$, $\rho_0$ determines two distinct representations on $V_{-1},V_0$. Thus in general we achieve $n+m$ number of distinct eigenvalues of $\rho^\text{gen}(L) = \rho_0L_0$, where $n+m = \operatorname{dim}(V_{-1}\oplus V_0) = \operatorname{dim}V_{-1} + \operatorname{dim}V_0$. On the other hand, when $t=1$ and $\partial =1$, the components $\rho_0^{1}=\rho_0^0$ of $\rho_0$ must coincide. If $L_0=L_{-1}$ (see {\bf Corollary \ref{ident-t}}), we have two copies of the $n$ eigenvalues of $\rho(L)$. In this $t=1$ case, we can say that our notion of 2-graded integrability consist of "two copies" of the usual integrability.

\medskip

\begin{remark}\label{homotopyinvariants}
    Here we make the crucial comment that the genuine matrix representation $\rho^\text{gen}$ is only available for \textit{strict} Lie 2-algebra 2-representations. The above notion of "eigenvalues" breaks down in the general $L_\infty$ case, as invariance of class functions is defined only up to homotopy --- ie. they are equipped with a chain homotopy $_2\operatorname{ad}^* \Rightarrow \operatorname{id}$ trivializing the 2-coadjoint action. Here, the data of this homotopy is in fact a crucial part of the very notion of "conserved quantity", and the "eigenvalues" should be replaced by some homotopy categorical construction.
\end{remark}

\subsection{2-Kirillov-Kostant $R$-bracket on $C^\infty(\g^*[1])$}\label{2kkbracket}
We first generalize the standard Kirillov-Kostant Poisson structure to the Lie 2-algebra context. This shall serve as the appropriate setting for constructing a canonical 2-Lax pair on the delooping space $B(\g^*[1])$ of the dual complex  of a given Lie 2-bialgebra $(\g;R^\wedge)$.

Recall from \S \ref{2grpois2alg} that the 1-forms $\Omega_{B(\g^*[1])}^1$ are generated by homogeneous elements of the form $d\phi$ for dg smooth functions $\phi\in C(B(\g^*[1]))=\operatorname{CE}(\g^*[1])$.  By extending the canonical evaluation pairing,  $$\langle-,-\rangle: \g^{\otimes n} \odot (\g^*[1])^{\otimes n}\to \mathbb{C}[1]\,,$$ there is a non-degenerate invariant bilinear form which identifies the 1-vectors, defined as the dual $T_{B(\g^*[1])}=\operatorname{Hom}_{dg}(\Omega^1_{B(\g^*[1])},\bbC)$ to the 1-forms, with the dg smooth functions.

Given $\g$ is a Lie 2-bialgebra with a compatible graded cobracket $\delta$, for $\phi,\phi'\in  C(B(\g^*[1]))$ we can define $\{\phi,\phi'\}^*\in \operatorname{CE}(\g^*[1])$ to be the element corresponding to $[d\phi,d\phi']\in \Omega^\bullet_{B(\g^*[1])}$ under the evaluation pairing $\langle-,-\rangle$. We write this with the notation
\begin{equation}
     \{\phi,\phi'\}^*=\langle-,[d\phi,d\phi']\rangle  \label{2poisdual}.
\end{equation}
\begin{proposition}\label{prop:2poismain}
For a Lie 2-bialgebra $(\g;\delta)$, the Chevalley-Eilenberg complex $\operatorname{CE}(\g^*[1])$, equipped with the Poisson bracket $\{-,-\}^*$ given in \eqref{2poisdual}, is a dg Poisson algebra. In other words, $B(\g^*[1])$ is a 1-shifted Poisson dg manifold.
\end{proposition}
\begin{proof}
The proof runs analogously to the ordinary Lie algebra case. The graded Jacobi identity clearly follows from that of $[-,-]$ on $\g$. Thus let us verify the graded Leibniz rule 
\begin{equation*}
    Q\{\phi,\phi'\}^* =\{Q\phi,\phi'\}^*+(-1)^{|\phi|} \{\phi,Q\phi'\}^*
\end{equation*}
with respect to the Chevalley-Eilenberg differential $Q = \delta-D_t$ given by the Lie 2-algebra cobracket $\delta$ on $\g$. 

By direct computation from \eqref{2poisdual}, we have 
\begin{equation*}
    Q\{\phi,\phi'\}^*  =\langle -,\delta [d\phi,d\phi']\rangle  -\langle -,D_t[d\phi,d\phi']\rangle 
\end{equation*}
and since $Q$ commutes with $d$,\footnote{Here we have abused notation to write $d\phi\otimes 1+1\otimes d\phi$ still as $d\phi$.} 
\begin{align*}
    \{Q\phi,\phi'\}^*+(-1)^{|\phi|} \{\phi,Q\phi'\}^* &= \langle -,[d(Q\phi),d\phi']\rangle +(-1)^{|\phi|}\langle -,[d\phi,d(Q\phi')]\rangle \\ 
    &= \langle -,[Q(d\phi),d\phi']\rangle +(-1)^{|\phi|}\langle -,[d\phi,Q(d\phi')]\rangle \\ 
    &= \langle -, [\delta(d\phi),d\phi']\rangle +(-1)^{|\phi|}\langle -,[d\phi,\delta(d\phi')]\rangle\\
    &\qquad - \langle -,[D_{t}(d\phi),d\phi']\rangle -(-1)^{\tilde \phi}\langle -,[d\phi,D_{t}(d\phi')]\rangle\, \\
\end{align*}
where $\tilde\phi$ denotes the cohomological parity of $\phi$ in terms of the odd-even decomposition \eqref{paritydecomp}. We therefore see that the compatibility between $\{-,-\}^*$ and $Q$ are precisely the conditions \eqref{pfeif} and the cocycle conditions \eqref{2cocycoh}.

\end{proof}

We now construct an alternative Poisson structure on $B(\g^*[1])$ by explicitly making use of the classical 2-$r$-matrix.  In analogy with  \eqref{algmap}, we first define a map $\varphi = ({\varphi}_{-1},{\varphi}_0):\mathfrak{g}\rightarrow\mathfrak{g}$ of 2-graded vector spaces, then use it to define an alternative $L_2$-bracket $[-,-]_R$ on $\g$. Let us fix the bases $\{T_i\}_i,\{S_a\}_a$ of $\mathfrak{g}_0,\mathfrak{g}_{-1}$ respectively.

\begin{proposition}\label{dgmap}
Given an element $R\in\g_0\otimes\g_{-1}\oplus\g_{-1}\otimes\g_0$ of degree -1 in $\g^{\otimes 2}$, the map $\varphi = ({\varphi}_{-1},{\varphi}_0):\mathfrak{g}\rightarrow\mathfrak{g}$ defined by
\begin{eqnarray}
     \varphi_{-1}:\mathfrak{g}_{-1} \rightarrow\mathfrak{g}_{-1},&\qquad& Y\mapsto 
      (R^\wedge)^{ia}\langle Y,T_i\rangle S_a,\nonumber\\
      \varphi_0:\mathfrak{g}_0\rightarrow\mathfrak{g}_0,&\qquad& X\mapsto (R^\wedge)^{ai}\langle X,S_a\rangle T_i,\nonumber 
\end{eqnarray}
is a 2-vector space homomorphism (ie. a dg linear map) if and only if  $D_t^-R^\wedge=0$.
\end{proposition}
\begin{proof}
Clearly, $\varphi$ is linear, hence it remains to show that  $t\varphi_{-1}= \varphi_0t$. By definition, this requires
\begin{equation}
    (R^\wedge)^{ia}T_i\wedge t(S_a) = (R^\wedge)^{ai} t(S_a)\wedge T_i \nonumber
\end{equation}
for each basis elements $T_i\in\mathfrak{g}_0,S_a\in\mathfrak{g}_{-1}$. In other words, the combination $(R^\wedge)^{ia}t_a^j$ is skew-symmetric; this is precisely the condition $D_t^-R^\wedge=0$ in  \eqref{2cybe} \cite{Bai_2013}.
\end{proof}

\begin{proposition}
Suppose $R\in\g_0\otimes\g_{-1}\oplus\g_{-1}\otimes\g_0$ satisfies the 2-CYBE \eqref{2cybe}. The bracket defined by
    \begin{equation*}
    [Y+X,Y’+X’]_R = [\varphi(Y+X),Y’+X’] + [Y+X,\varphi(Y’+X’)],
\end{equation*}
is a Lie 2-algebra bracket which satisfies 
\begin{equation}
    [g_X+f_X,g'_Y+f'_X] = \langle-, [Y+X,Y'+X']_R\rangle \label{Rbracket}
\end{equation}
where $f^{(\prime)}_X=\langle -,X^{(\prime)}\rangle\in \g^*_{-1}$ and $g_Y^{(\prime)} = \langle-, Y^{(\prime)}\rangle\in \g^*_0$.
\end{proposition}
\begin{proof}
    Recall \cite{Bai_2013} that the skew-symmetric piece $R^\wedge$ of a solution $R$ to the 2-CYBE \eqref{2cybe} defines the cobracket $dR^\wedge(Y+X)=\delta(Y+X) =\delta_{-1}(Y)+\delta_0(X)$ given by 
    \begin{equation*}
        \delta_{-1}(Y) = [Y\otimes 1+1\otimes Y,R^\wedge],\qquad \delta_0(X) = [X\otimes 1+1\otimes X,R^\wedge],
    \end{equation*}
    and the symmetric piece $R^\odot=\langle-,-\rangle$ defines a $_2\operatorname{ad}$-invariant pairing. These facts allow us to compute directly (cf. \cite{Meusburger:2021cxe}) that, for each basis element $Z_i=S_i+T_i\in\mathfrak{g}$,
    \begin{eqnarray*}
        [h,h'](Z_i) &=&  \langle h\otimes h',\delta(Z_i)\rangle = \langle h\otimes h',[Z_i\otimes 1+1\otimes Z_i,R^\wedge]\rangle \\
        &=& (R^\wedge)^{jk}\langle h\otimes h',[Z_i,Z_j]\otimes Z_k + Z_j\otimes [Z_i,Z_k]\rangle \\
        &=& (R^\wedge)^{jk}\left(R^\odot(Z,[Z_i,Z_j])R^\odot(Z',Z_k) + R^\odot(Z,Z_j)R^\odot(Z',[Z_i,Z_k])\right),\\
        &=& - (R^\wedge)^{jk}\left(R^\odot([Z,Z_j],Z_i)R^\odot(Z',Z_k)+R^\odot(Z,Z_j)R^\odot([Z',Z_k],Z_i)\right) \\
        &=& R^\odot([Z,\varphi(Z')],Z_i) + R^\odot([\varphi(Z),Z'],Z_i) \\
        &=& R^\odot([Z,Z']_R,Z_i) = \langle Z_i,[Z,Z']_R\rangle,
    \end{eqnarray*}
    where we  abbreviated the graded elements $h=g+f,h'=g’+f’\in\mathfrak{g}^*$ and used that $Z^{(')}\equiv \langle h^{(')}, \rangle\in \g$. This proves \eqref{Rbracket}.

    Now let us establish that $[-,-]_R$ is a genuine $L_2$-bracket on $\mathfrak{g}$. Since $[-,-]$ by hypothesis is equivariant and satisfies the Peiffer identity with respect to $t$, the fact that $\varphi$ is a 2-vector space homomorphism implies the same for $[-,-]_R$. It thus suffices to check the 2-Jacobi identities for $[-,-]_R$, but this directly follows from \eqref{Rbracket} (cf. \cite{kosmann1997integrability}),
    \begin{eqnarray*}
        \langle Z_0,\circlearrowright [[Z,Z']_R,Z'']_R\rangle =~ (\circlearrowright [[h,h'],h''])(Z_0) =0\qquad \forall Z_0\in\g.
    \end{eqnarray*}
\end{proof}

Extending $\varphi$ to a dg linear map on $\operatorname{CE}(\g^*[1])$ in an obvious way, the bracket $[-,-]_R$ extends correspondingly to $C(B(\g^*[1]))$.
\begin{lemma}\label{prop:2poismainbis}
Give a solution $R$ of the 2-CYBE, the Poisson bracket $\{-,-\}_R^*$, defined by the following formula
\begin{equation}
    \{\phi,\phi'\}_R^* = \langle -, [d\phi,d\phi']_R\rangle,\qquad \phi,\phi'\in  C(B(\g^*[1])) \label{2poisdualR},  
\end{equation}
is a 1-shifted Poisson structure on $B(\g^*[1])$.  We call this a {\bf 2-Kirillov-Kostant (2KK) $R$-bracket} on $C(\g^*[1])$. 
\end{lemma}
\begin{proof}
    This follows from the fact that $[-,-]_R$ is a $L_2$-bracket, hence the proof of \textbf{Proposition \ref{prop:2poismain}} applies. 
\end{proof}

\subsection{Canonical 2-Lax pair on $\g^*[1]$}\label{2laxconstruction}
Fix a $_2\operatorname{Ad}^*$-invariant Hamiltonian $H\in C(B(\g^*[1]))$ (as defined in Definition \ref{def:invariant}). We are now ready to finally construct a canonical 2-Lax pair $(L,P)$ on $(B(\g^*[1]),\{-,-\}_R^*,H)$ in this section according to \eqref{2grlax}, based on the 2-KK Poisson structure $\{-,-\}_R^*$ \eqref{2poisdualR} as well as the underlying classical 2-$r$-matrix. 

Following \cite{Meusburger:2021cxe}, we will specify $L,P$ linearly on the subspace $\g \subset C^\bullet(B(\g^*[1])) = \operatorname{CE}(\g^*[1])$. We take (note the subscripts denote the \textit{Lie 2-algebra} degree, not the total degree) 
\begin{equation*}
    L = L_{-1} + L_0,\qquad P=P_{-1} + P_0,
\end{equation*}
with
\begin{align}
    L_0&\in \g_0\otimes\mathfrak{g}_0,\qquad L_{-1} \in \g_{-1}\otimes \g_{-1},\nonumber\\
    P_{-1}&\in \g_0\otimes\mathfrak{g}_{-1},\qquad P_0\in \g_{-1}\otimes\mathfrak{g}_0\,.\nonumber
\end{align}
Notice $L$ has even total degree and $P$ has odd total degree in the complex \eqref{tensorxmod}. An incarnation of this splitting of the 2-Lax pairs will also appear in \S \ref{2laxzero} later.

To write down the explicit form of these maps $L,P \in C(\g^*[1])\otimes\g$, we will leverage the canonical degree-1 evaluation pairing $\langle-,-\rangle$ to identify $\g_0\cong \g_{-1}^*$ and $\g_{-1}\cong \g_0^*$. Fix bases $\{T_i\}_i,\{S_a\}_{a}$ of $\mathfrak{g}_0,\mathfrak{g}_{-1}$, and suppose the classical 2-$r$-matrix $R$ on $\g$ is invertible. We make use of a basic linear algebra fact \cite{LU2002119} that the inverse of an off-diagonal block matrix, such as $R$ where the off-diagonal pieces are given by $R_1,R_2$, is the off diagonal matrix with blocks  $R_2^{-1}$ and $R_1^{-1}$,
and hence the inverse of the symmetric piece $(R^\odot_1)_{ai}$, for instance, has matrix elements $ 
((R^\odot_2)^{-1})^{ai}$. We define, for $f\in \g_{-1}^*\cong \g_0$ and $g\in\g_0^*\cong \g_{-1}$, the following formulas
\begin{eqnarray}
    L_0= (R^\odot_2)^{ai}\langle -,S_a\rangle T_i: f\mapsto (R^\odot_2)^{ai}f(S_a)T_i,&\qquad& P_{-1} = \varphi_{-1}(dH): f\mapsto \langle f,\varphi(dH)\rangle,\nonumber\\
    L_{-1}=(R^\odot_1)^{ia}\langle-, T_i\rangle S_a: g\mapsto (R^\odot_1)^{ia}g(T_i)S_a,&\qquad& P_0 = \varphi_0(dH): g\mapsto \langle g,\varphi(dH)\rangle\,.\label{2lax}
\end{eqnarray}
We will now show that $(L,P):\mathfrak{g}^*[1]\rightarrow\mathfrak{g}$ is indeed a 2-Lax pair as in {\bf Definition \ref{2grlax}}. 

Following the notation in \S \ref{2grpois2alg}, we let $\textbf{t}$ denotes the differential on the sheaf $C(B(\g^*[1]))$ defining the dg manifold $B\g*[1].$ 
\begin{theorem}\label{construct2lax}
Let $(\g;R)$ denote a Lie 2-bialgebra equipped with a 2-graded classical $r$-matrix, and let $H\in C(B(\g^*[1]))$ denote an invariant Hamiltonian on the shifted dual complex $\g^*[1]$. Then $(L,P)$ given in \eqref{2lax} is a 2-Lax pair of the 2-graded Hamiltonian system $(B(\g^*[1]),\{-,-\}^*_R,H)$ for which the Lax potential $L$ satisfies
\begin{equation}
    (1\otimes t - \textbf{t}\otimes 1)L_{-1} = 0,\qquad \{L,L\}^*_R = [L\otimes 1 + 1\otimes L,R^\wedge], \label{2laxcond}
\end{equation}
where we have extended the $L_2$-bracket $[-,-]$ to $\mathfrak{g}^{\otimes 2}$.
\end{theorem}
\begin{proof}
First we compute the coefficients $$(dL_{-1})^i = R^\odot_2{}^{bi}S_b,\qquad (dL_0)^a = R^\odot_1{}^{ja}T_j.$$ We note also that the $_2\operatorname{Ad}^*$-invariance of $H$ \eqref{adinv} implies, in particular, that
\begin{equation*}
    [ {Y+X},dH]=0,\qquad  \forall\, Y\in\g_{-1},\, \forall X\in \g_0,
\end{equation*}
 (we emphasize we use the bracket $[-,-]$ here and not $[-,-]_R$). 

Then from the 2-KK Poisson structure \eqref{2poisdualR} we have
\begin{eqnarray}
    \{H,L\}_R^* &=& \langle -,[dH,dL^{i,a}]_R\rangle (T_i\oplus S_a) \nonumber\\
    &=&  \langle -,[dH,dL^i_{-1}]_R\rangle T_i+ \langle -,[dH,dL_0^a]_R\rangle S_a \nonumber\\
    &=& \langle -,[\varphi_{-1}(dH),dL_{-1}^i] + [ dH,\varphi_{-1}(dL^i_{-1})]\rangle T_i \nonumber \\
    &\qquad& +~  \langle -,[\varphi_0(dH),dL_0^a] + [ dH,\varphi_0(dL_0^a)]\rangle S_a \nonumber\quad \textrm{\tiny use invariance of Hamiltonian}\\
    &=& \langle -,[\varphi_{-1}(dH),R^\odot_2{}^{bi}S_b]\rangle T_i +  \langle -,[\varphi_0(dH),R^\odot_1{}^{ja}T_j] \rangle S_a \nonumber\\
    &=& -R^\odot_2{}^{bi}\langle -,S_b\rangle [\varphi_{-1}(dH),T_i] - R^\odot_1{}^{ja}\langle -,T_i\rangle [\varphi_0(dH),S_a] \nonumber\\
    &=& [L,P],\nonumber
\end{eqnarray}
where we have used the invariance of the 2-Casimir $R^\odot$. This proves the first statement. 

To prove the second statement, we first note that we have the following expressions
\begin{equation*}
    \varphi_{-1}(S_a) = (R^\odot_1)_{ai}(R^\wedge_2{})^{ic}S_c,\qquad \varphi_0(T_i) = (R^\odot_2)_{ib}(R^\wedge_1{})^{bj}T_j\nonumber
\end{equation*}
for $\varphi$. Hence by a direct computation,
\begin{eqnarray}
    \{L,L\}_R^* &=& \{L^{a,i},L^{b,j}\}^*(S_a+T_i)\otimes (S_b+T_j) \nonumber\\
    &=& \langle -,[dL^{a,i},dL^{b,j}]_R\rangle(S_a+T_i)\otimes (S_b+T_j)\nonumber\\
    &=& (R^\odot_2{}^{ai'}+R^\odot_1{}^{ia'})(R^\odot_2{}^{bj'}+R^\odot_1{}^{jb'})\langle -,[\varphi_0 T_{i'}+ \varphi_{-1} S_{a'},T_{j'}+S_{b'}] \nonumber\\
    &\qquad&+~[T_{i'}+S_{a'},\varphi_0 T_{j'} + \varphi_{-1} S_{b'}]\rangle (S_a+T_i)\otimes (S_b+T_j) \nonumber\\
    &=& (R^\wedge_1{}^{al} + R^\wedge_2{}^{ic})(R^\odot_2{}^{bj'}+R^\odot_1{}^{jb'})\langle -,[T_l+S_c,T_{j'}+S_{b'}]\rangle ((S_a+T_i)\otimes (S_b+T_j)) \nonumber\\
    &\qquad& + (R^\odot_2{}^{ai'}+R^\odot_1{}^{ib'})(R^\wedge_1{}^{bm} + R^\wedge_2{}^{jd})\langle -,[T_{i'}+S_{a'},T_m + S_d] \rangle ((S_a+T_i)\otimes (S_b+T_j)) \nonumber\\
    &=& - (R^\wedge_1{}^{al} + R^\wedge_2{}^{ic})(R^\odot_2{}^{bj'}+R^\odot_1{}^{jb'})\langle -,T_{j'}+S_{b'}\rangle ((S_a+T_i)\otimes [T_l+S_c,S_b+T_j]) \nonumber\\
    &\qquad& - (R^\odot_2{}^{ai'}+R^\odot_1{}^{ib'})(R^\wedge_1{}^{bm} + R^\wedge_2{}^{jd})\langle -,T_{i'}+S_{a'} \rangle ([S_a+T_i,T_m + S_d]\otimes (S_b+T_j)) \nonumber\\
    &=& - (R^\wedge_1{}^{al} + R^\wedge_2{}^{ic})((S_a+T_i)\otimes [T_l+S_c,L]) \nonumber\\
    &\qquad& - (R^\wedge_1{}^{bm} + R^\wedge_2{}^{jd})([L,T_m + S_d]\otimes (S_b+T_j)) \nonumber\\
    &=& -[r,L\otimes 1 + 1\otimes L].
\end{eqnarray}
Now note that the differential $\textbf{t}$ on $B(\g^*[1])$ acting at degree 1 is simply the Lie 2-algebra differential $t$. Whence we have, finally,
\begin{eqnarray*}
    (\textbf{t}\otimes 1)L_{-1} &=&  (R^\odot_2)^{ai}\langle t-, T_i\rangle S_a =(R^\odot_2)^{ai}\langle t-, S_a\rangle T_i  \nonumber\\
    &=& (R^\odot_1)^{ia}\langle -,tS_a\rangle T_i =(R^\odot_1)^{ia}\langle -,T_i\rangle tS_a = (1\otimes t)L_0\nonumber    
\end{eqnarray*}
as desired, where we have used the symmetry, as well as the property \eqref{tsymm} of the quadratic 2-Casiimr $K=R^\odot$.
\end{proof}
\noindent The special properties that the 2-Lax potential $L$ satisfies in this case allows us to prove the following.

\begin{corollary}\label{induced1lax} 
    Given the hypotheses of \textbf{Theorem \ref{construct2lax}}, the canonical 2-Lax pair \eqref{2lax} induces an ordinary Lax pair $({\underline{L}},\underline{P})\in C(\g^*_0)\otimes \g_0$  on the Hamiltonian system $(\g_0^*,\{-,-\}_0^*,H)$ at degree-0.
\end{corollary}
\begin{proof}
    Recall that we have extended the $t$-map to act on all of $\g$, such that $t(Y+X) = tY$ for each $X\in\g_0,Y\in\g_{-1}$ \cite{Chen:2012gz}. Whence we have, for instance, $(1\otimes t)L = (1\otimes t)L_{-1}$.
    
    First, let us apply the $t$-map on the values $\g$ of $L$ to \eqref{2grlax}. This gives (neglecting the pullback $\gamma^*$)
    \begin{align*}
        (1\otimes t)\dot L &= (1\otimes t)([L,P]) =[L_0,(1\otimes t)P_{-1}] +[(1\otimes t)L_{-1},P_0],
    \end{align*}
    where we have used the Leibniz rule for $t$. By applying $ \textbf{t}\otimes 1 $, we see that the first term vanishes due to degree reasons, hence
    \begin{align*}
        (\textbf{t}\otimes t)\dot{L}_{-1}=(\textbf{t}\otimes 1)[(1\otimes t)L_{-1},P_0]=[(\textbf{t}\otimes t)L_{-1},(\textbf{t}\otimes 1)P_0]
    \end{align*}
    by the fact that $\textbf{t}$ is a map of dgca's.

    We now prove that the dynamics $\dot{\underline{L}}$ is generated by the restriction $H|_{B\g^*_0}$. By applying $(\textbf{t}\otimes 1)$ to \eqref{2grlax} and go through the same computation, we have\footnote{Note the pullback $\gamma^*$ is a map of dgca's, and hence commutes with the differential $\textbf{t}$.}
    \begin{equation*}
        (\textbf{t}\otimes 1)\dot L = (\textbf{t}\otimes 1)\{H,L\}_R^* =\{(\textbf{t}\otimes 1)H,L\}_{R}^* + \{H,(\textbf{t}\otimes 1)L_{-1}\}_{R}^*,
    \end{equation*}
    where we have used the Leibniz rule of $\textbf{t}$ with respect to the 2-KK Poisson structure \eqref{2poisdualR}. Here, we apply instead $ 1\otimes t $, from which we see that the first term vanishes by degree reasons, whence we have
    $$(\textbf{t}\otimes t)\dot L_{-1} =  \{H,(\textbf{t}\otimes t)L_{-1}\}_{R}^*$$
    as $t$ does not act on $H$.

    These two computations imply that we obtain well-defined Lax operators $\underline{{L}} = (\textbf{t}\otimes t){L},\,\,\underline{P} = (\textbf{t}\otimes 1)P_0$ satisfying
    \begin{equation*}
        \dot{\underline{L}} = \{H,\underline{L}\}_{R}^* = [\underline{L},\underline{P}],
    \end{equation*}
    completing the proof.   
\end{proof}
\noindent In other words, the projection to degree 0 defines a functor $\mathsf{Lie2Bialg}\to \mathsf{LieBialg}$ which sends canonical 2-Lax pairs \eqref{2lax} to canonical Lax pairs \eqref{lax}.

Now it is known \cite{Bai_2013} that this functor admits a left adjoint $\mathsf{LieBialg}\to \mathsf{Lie2Bialg}$ given by the trivial Lie 2-bialgebra $\mathsf{g}\mapsto \big(\operatorname{id}_\mathsf{g}=\mathsf{g}\xrightarrow{t=\id} \mathsf{g}\big)$. We thus have the following.

\begin{proposition}\label{ident-t}
Let $\mathsf{g}$ be a Lie (1-)bialgebra (with a classical $r$-matrix). If $(\hat L,\hat P)\in C(\mathsf{g}^*)\otimes\mathsf{g}$ form a (canonical) Lax pair on $\mathsf{g}^*$ (as in \eqref{lax}), then the following graded functions $L=\hat L\oplus \hat L,\,P=\hat P\oplus \hat P$ consisting of two copies of the original Lax pair, is a (canonical) 2-Lax pair on $C(B(\g^*[1]))$ where $\g=\operatorname{id}_\mathsf{g}=\operatorname{id}_{\mathsf{g}}:\mathsf{g}\xrightarrow{\id}\mathsf{g}$.
\end{proposition}

We will now switch gears a bit in \S \ref{zero2curv} in order to introduce a higher-homotopy generalization $\widehat{\Sigma_s\g}$ of the Kac-Moody algebra $\widehat{\g_k}$ described in \S \ref{wzw}. This is to setup the contents of \S \ref{2lax2km}, where we develop an application of the theory of 2-graded integrability \eqref{2grlax} above to a 3d topological-holomorphic field theory studied in \cite{Chen:2024axr}.

\section{An affine Kac-Moody 2-algebra}\label{zero2curv}
We begin with a \textit{balanced} (see \textbf{Definition \ref{balanced}}) Lie 2-algebra $\g=\g_{-1}\xrightarrow{t}\g_0$ over $\bbC$ \cite{Zucchini:2021bnn}, meaning that there is a degree-1 non-degenerate pairing $\langle-,-\rangle: \g_{-1}\odot\g_0\rightarrow\bbC$ which is invariant
\begin{equation*}
    \langle [X,X'],Y\rangle = -\langle X',X\rhd Y\rangle,\qquad \langle Y,tY'\rangle=\langle tY,Y'\rangle
\end{equation*}
for all $X,X'\in\g_0$ and $Y,Y'\in\g_{-1}$. This balanced property is the differential graded analogue of having an invariant pairing form for ordinary Lie algebras. 

In general, our "Kac-Moody 2-algebra" is a certain centrally-extended differential-graded Lie algebra (dgla) $\widehat{\mathcal{A}_s\g}$, where $s$ is a 2-cocycle on the tensor complex $\mathcal{A}\g=\mathcal{A}\otimes\g$ formed by a differential-graded commutative algebra (dgca) $\mathcal{A}$ with the Lie 2-algebra $\g =\g_{-1}\to\g_0$. Depending on the choice of $\mathcal{A}$, different Kac-Moody 2-algebras can arise (cf. \textit{Remark \ref{raviolo}}).

\subsection{Definition}\label{2kmdefinition}
As a proof of concept, we shall explicitly construct $\widehat{\mathcal{A}_s\g}$ under the "purely topological" choice given by the smooth de Rham complex $\mathcal{A}=\Omega^\bullet(\Sigma)$ on a 2d oriented surface $\Sigma$.\footnote{One says that  $\Omega^\bullet(\Sigma)$ is "purely topological" because it resolves the sheaf of locally-constant functions on $\Sigma$. As we will explain in  \S \ref{2laxintheift}, depending on the differential constraints appearing in the model under study, other more interesting sheaves of dgca's can be chosen.} 

Consider the following cochain complex
\begin{equation*}
    \Omega^{\bullet}(\Sigma,\g) = \Sigma\g_{-1}\xrightarrow{\hat d} \Sigma \g_0 \xrightarrow{\hat d} \Sigma\g_1 \xrightarrow{\hat d} \Sigma\g_2\,,
\end{equation*}
where $ \hat d = d-t$ and $\Sigma\g_p,\, p=-1,0,1,2$ consists of the graded pieces of total degree $p,$
\begin{align*}
    & \Sigma\g_{-1} = \Omega^0(\Sigma,\g_{-1}), \\
    &\Sigma\g_0 = \Omega^0(\Sigma,\g_0)\oplus  \Omega^1(\Sigma,\g_{-1}),\\
    &\Sigma\g_1 = \Omega^2(\Sigma,\g_{-1}) \oplus \Omega^1(\Sigma,\g_0),\\ 
    &\Sigma\g_2 = \Omega^2(\Sigma,\g_0).
\end{align*}
Denote by 
\begin{equation*}
    \Sigma\g^+ = \Sigma\g_0\oplus\Sigma\g_2,\qquad {\Sigma}\g^- = \Sigma\g_{-1}\oplus\Sigma\g_1
\end{equation*}
the pieces of the cochain complex $\Omega^{\bullet}(\Sigma,\g)$ with even/odd grading.\footnote{Of course, dgca's $\mathcal{A}$ in general do not terminate, but the splitting $\mathcal{A} =\mathcal{A}_+\oplus\mathcal{A}_-$ against the grading parity can always be performed.}

Analogous to the ordinary Kac-Moody algebra \eqref{kmalg}, let us introduce the following functions for 
\begin{align*}
    s^\lambda: \Sigma\g^+ \otimes \Sigma\g^+  \to \mathbb{C},\qquad s_-^\lambda: \Sigma\g^-\otimes\Sigma\g^-\to \mathbb{C}
\end{align*}
for a tuple of scalars $(\lambda_0,\lambda_1)\in\mathbb{C}$, we put
\begin{align}
    s^\lambda_+(\chi,\eta) &= \lambda\cdot \langle\langle  \chi, \hat d\eta\rangle\rangle \nonumber\\
    &= \lambda_1\int_\Sigma (\langle X, dy\rangle + \langle x, dY\rangle) - \lambda_0 \int_\Sigma \langle X, tY\rangle\,, \nonumber\\
    s^\lambda_-(\bar\chi,\bar\eta) &= \lambda\cdot \langle\langle \bar\chi, \hat \eta\rangle\rangle\nonumber\\
    & = \lambda_0\int_\Sigma(\langle  \bar X, d \bar y\rangle - \langle \bar x,d\bar Y\rangle) - \lambda_1\int_\Sigma\langle  \bar x, t\bar{ \mathsf{Y}}\rangle + \langle\bar{\mathsf{X}},t\bar{y}\rangle\,,\label{central}
\end{align}
with other components trivial (in particular, $s^\lambda$ is trivial on $\Sigma\g_2$), where 
\begin{equation}
    \chi=(\LaTeXunderbrace{x,X}_{\in\Sigma\g_0},\LaTeXunderbrace{\mathsf{X}}_{\in\Sigma\g_2})\in \Sigma\g^+ ,\qquad \bar\chi = (\LaTeXunderbrace{\bar x}_{\in\Sigma\g_{-1}},\LaTeXunderbrace{\bar X,\bar{\mathsf{X}}}_{\in\Sigma\g_1})\in \Sigma\g^-\,.
\end{equation}
We will use this notation to label elements in $\Sigma\g^+,\Sigma\g^-$, unless otherwise stated.

\medskip

By an integration by parts, 
it can be seen that the functions $s^\lambda,s_-^\lambda$ are graded skew-symmetric, but in opposite ways
\begin{equation}
    s^\lambda_+(\mathcal{X},\mathcal{Y}) = (-1)^{|\mathcal{X}|+|\mathcal{Y}|+1} s^\lambda_+(\mathcal{Y},\mathcal{X}),\qquad s_-^\lambda(\bar{\mathcal{X}},\bar{\mathcal{Y}}) = (-1)^{|\bar{\mathcal{X}}|+|\bar{\mathcal{Y}}|} s_-^\lambda(\bar{\mathcal{Y}},\bar{\mathcal{X}})\,,\label{centralgradedsymm}
\end{equation}
where $|\mathcal{X}| = (|x|,|X|) = (0,1)$ counts the form degrees. By denoting the $\mathbb{Z}_2$-parity of the underlying forms with $p=\pm$, we can then view $s^\lambda_+= s_+^{\lambda_0|\lambda_1}: \Sigma\g^+ \otimes \Sigma\g^+  \to \mathbb{C}^{1|1}$ as a function which takes values in the complex \textit{superline} $\mathbb{C}^{1|1}$; similarly for $s_-^\lambda = s_-^{\lambda_0|\lambda_1}$. We identify the scalars $\lambda_0\in\mathbb{C}^{1|0}$ as even and $\lambda_1\in\mathbb{C}^{0|1}$ as odd. 

\medskip

Consider the spaces $\Sigma\g^+\oplus s_+\bbC^{1|1}[1],\,\, \Sigma\g^-\oplus s_-\bbC^{1|1}[1]$ with $\mathbb{Z}_2$-graded central pieces generated by $s_+,s_-$, respectively, in cohmological degree -1.\footnote{This  degree essentially comes from the degree of the balancing $\langle-,-\rangle:\g^{\otimes 2}\to \mathbb{C}[1]$.}  Define the space $$\widehat{\Sigma_s\g} = \big(\Sigma\g^+\oplus s_+\bbC^{1|1}[1]\big) \oplus  \big(\Sigma\g^-\oplus s_-\bbC^{1|1}[1]\big) $$ as their direct sum. In the following, we will endow this space with the structure of a Lie 2-\textit{super}algebra.

\paragraph{The graded brackets.}
For $p=\pm$, we now introduce the 2-graded bracket structure on $\widehat{\Sigma_s\g}$,
$$[-,-]: \widehat{\Sigma_s\g}\wedge \widehat{\Sigma_s\g}\to \widehat{\Sigma_s\g}\,,$$
induced by the underlying Lie 2-algebra $\g$. For $\chi,\eta\in\Sigma\g^+$ and $\bar\chi,\bar\eta\in\Sigma\g^-$ we have
\begin{align}
    [\chi,\eta] &= \big([x,y],[X,y] + [x,Y],[x,\mathsf{Y}]+[\mathsf{X},y]\big)+ s^{\lambda_0|\lambda_1}_+(\mathcal{X},\mathcal{Y})  &\in \Sigma\g^+\oplus \bbC^{1|1} \nonumber\\
    [\bar{\chi},\bar{\eta}] &=  (0,[\bar x,\bar Y] - [\bar X,\bar y],[\bar X,\bar Y])  + s_-^{\lambda_0|\lambda_1}(\bar{\mathcal{X}},\bar{\mathcal{Y}}) &\in \Sigma\g^+\oplus \bbC^{1|1}\,.\label{bracs}
\end{align}
Note the $\mathbb{Z}_2$-graded skew-symmetry with respect to the parity of the form degrees is manifest,
$$[\chi,\eta] = (-1)^{|\chi|+|\eta|+1}[\eta,\chi],\qquad [\bar\chi,\bar\eta] = (-1)^{|\bar\chi|+|\bar\eta|}[\bar\eta,\bar\chi]\,,$$
where $|\chi| = (|x|,|X|,|\mathsf{X}|) = (0,1,2)$; see also \eqref{centralgradedsymm}). This also makes the graded Jacobi identities manifest, as follows from \eqref{2jacob}.

\medskip

One may also introduce the brackets induced by the degree -1 crossed-module bracket $[-,-]^\ovo: \g_{-1}\wedge\g_{-1}\to\g_{-1}$ defined by $[-,-]^\ovo = [t-,-]$ through the Peiffer identity \eqref{pfeif}. However, these are in fact \textit{odd} brackets,
\begin{align}
   [\chi,\eta]^\ovo&= (0,0,[X,Y]^\ovo)  &\in \Sigma\g^- \nonumber\\
   [\bar\chi,\bar\eta]^\ovo &=  \big([\bar x,\bar y]^\ovo, 0,[\bar x,\bar{\mathsf{Y}}]^\ovo + [\bar{\mathsf{X}},\bar y]^\ovo\big) &\in\Sigma\g^-,\label{bracs2}
\end{align}
which sends elements in $\Sigma\g^+$ to the dual sector $\Sigma\g^-$. In other words, the brackets $[-,-]^\ovo$ \eqref{bracs2} themselves have a parity degree attached to it. We shall think of them as separate from \eqref{bracs}.

\paragraph{The pairing form.}

Now one may wonder about the mixed even-odd brackets, which takes for instance the input $\chi,\bar\eta$ from different components of $\widehat{\Sigma_s\g}$. For this, we will induce them from a pairing form such that  the even/odd subspaces $\Sigma\g,\Sigma\g^-$ are by definition \textbf{coisotropic} in  $\widehat{\Sigma_s\g}$.
 
We define the integration pairing form $\llparenthesis -,-\rrparenthesis$ on $\widehat{\Sigma_s\g}$, which is odd in the total degree,
\begin{align}
    \llparenthesis \chi\oplus c\,|\, \bar{\chi}\oplus \bar{c}\, \rrparenthesis &= \langle\langle \chi\,|\,\bar{\chi}\rangle\rangle + (c|\bar c) \nonumber\\
    &=  \int_\Sigma \langle X,\bar X\rangle + \langle x,\bar{\mathsf{X}}\rangle + \langle\mathsf{X},\bar x\rangle + (c^*_1\bar c_1 + c_0^*\bar c_0)\,, \label{2kmpairing}
\end{align}
where $(-|-)$ is the canonical inner product on the superline $\bbC^{1|1}$ and $\langle-,-\rangle$ is the (cohomological) degree 1 pairing underlying the balanced Lie 2-algebra $\g$. 

Note the "bar" on top of $\bar c$ does not mean complex conjugation, but instead refers to the total degree parity. Clearly, the pairing form \eqref{2kmpairing} is graded symmetric
\begin{equation*}
    \llparenthesis \chi\oplus c\,|\,\bar{\chi}\oplus \bar c\, \rrparenthesis_0 = (-1)^{|\chi|+|\bar\chi|}(-1)^{|c|+|\bar c|} \llparenthesis \bar \chi  \oplus \bar c\,|\, \chi \oplus  c\rrparenthesis_0\,
\end{equation*}
under the parity of the form degrees $|\chi| = (|x|,|X|,|\mathsf{X}|) = (0,1,2)$.

\paragraph{The coadjoint action.} 
Using \eqref{2kmpairing} and \eqref{bracs}, we define the coadjoint actions by
\begin{align}
    \llparenthesis [\chi_1,\chi_2]\,|\,\bar{\chi}_3\oplus \bar c_3 \rrparenthesis &= (-1)^{|\chi_1|+|\chi_2|+1}\llparenthesis \chi_2 \,|\, \operatorname{ad}^*_{\chi_1}(\bar{\chi}_3\oplus \bar c_3)\rrparenthesis\,, \nonumber\\ 
    \llparenthesis \bar\chi_1\oplus \bar c_1 \,|\,[\bar{\chi}_2,\bar{\chi}_3] \rrparenthesis &= (-1)^{|\bar\chi_2|+|\bar\chi_3|}\llparenthesis \operatorname{ad}^*_{\bar\chi_2}(\chi_1\oplus c_1) \,|\, \bar{\chi}_3\rrparenthesis\,;\label{coadjointactions}
\end{align}
recall $[\bar\chi,\bar\eta]$ lies in the even sector $\Sigma\g$. By properties of the Hodge star, we perform a direct computation to deduce that
\begin{align}
    \operatorname{ad}^{*}_{\chi_1}(\bar\chi_3\oplus \bar c_3) &= \big([x_1,\bar x_3],[x_1,\bar{X}_3]+ \lambda_1(\bar c_3)_1^*dx_1- \lambda_0(\bar c_3)_0^*tX_1, \nonumber\\
    &\qquad \qquad  [x_1,\bar {\mathsf{X}}_3] + [X_1,\bar X_3] + [\mathsf{X}_1,\bar x_3] -\lambda_1(\bar c_3)_1^*dX_1\big) + 0, \label{evencoad1}\\
    \operatorname{ad}^*_{\bar \chi_2}(\chi_1\oplus c_1) &= \big(0,[\bar X_2, x_1]+ [\bar x_2, X_1]+(\bar c_1)_0^*\lambda_0 d\bar x_2,\nonumber\\
    &\qquad\qquad [\bar X_2, X_1] + (\bar c_1)_0^*\lambda_0d\bar X_2-(\bar c_1)_1^*\lambda_1t\bar x_2\operatorname{vol}_\Sigma\big)+0,\label{evencoad2}
\end{align}
where the trailing "$+0$" indicates the absence of central terms in the coadjoint action.

Similarly, from the crossed-module brackets \eqref{bracs2} we can also compute
\begin{align*}
    \,^\ovo\operatorname{ad}^*_{\chi_1}(\chi_3\oplus c_3) &= \big(0,[tX_1,x_3],0\big)+0, \\
    \,^\ovo\operatorname{ad}^*_{\bar \chi_2}(\chi_1\oplus c_1) &=([t\bar x_2,x_1],0,[t\bar x_2,\mathsf{X}_1]+ [t\bar{\mathsf{X}}_2,x_1]\big)+0\,.  \label{oddcoad}
\end{align*}
By the equivariance of $t$, \eqref{pfeif}, these brackets in fact all lie in the image of $t$,
\begin{align*}
    \,^\ovo\operatorname{ad}^*_{\chi_1}(\chi_3\oplus c_3)&= \big(0,t[X_1,x_3],0\big)+0 \\ 
    \,^\ovo\operatorname{ad}^*_{\bar \chi_2}(\chi_1\oplus c_1)&=(t[\bar x_2,x_1],0,t[\bar x_2,\mathsf{X}_1]+ t[\bar{\mathsf{X}}_2,x_1]\big)+0 \,.
\end{align*}
These contributions therefore vanish if one projects onto cohomology.

\subsection{Kac-Moody 2-algebra as a dg central-extension}\label{2kmproofsec}
Let us now put all of the above structures together. Recall the differential $\hat d=d-t$ on $\Omega^\bullet(\Sigma,\g)=\Sigma\g^+\oplus \Sigma\g^-$ is off-diagonal, and is given by
\begin{equation*}
    \hat d: \begin{pmatrix}
        \chi \\ 
        \bar\chi
    \end{pmatrix}\mapsto\begin{pmatrix}
(- t\bar x,d\bar x,d\bar X - t\bar{\mathsf{X}})\\
                (0,dx-tX,dX)\,
    \end{pmatrix}
\end{equation*}
where $t$ acts trivially on $\g_0$ and $d$ acts trivially on $\Omega^2(\Sigma)$. We now prove the main result of this section.
\begin{theorem}\label{2kmproof}
    Put $s=s_++s_-$. The \textbf{Kac-Moody 2-algebra $(\widehat{\Sigma_s\g},\hat d)$} defined in \S \ref{2kmdefinition}, with the bracket structures \eqref{bracs}, \eqref{evencoad1}, \eqref{evencoad2}, is a dg Lie algebra central extension $$s\bbC^{1|1}[1]\rightarrow \widehat{\Sigma_s\g}\rightarrow\Sigma\g^+\oplus \Sigma\g^-,$$ where $\bbC^{1|1}[1] = \bbC^{1|1}\xrightarrow{0}0$ is the shifted superline $\mathbb{C}^{1|1} = \bbC^{1|0} \oplus\bbC^{0|1}$ as a trivial $\Sigma\g$-module. 
    Moreover, $\widehat{\Sigma_s\g}^\pm$ are coisotropic under the pairing \eqref{2kmpairing}.
\end{theorem}
\begin{proof}
The graded skew-symmetry of the brackets \eqref{bracs} is clear. We proceed by first proving the following lemma.
    \begin{lemma}\label{tequiv}
    $\hat d=d-t$ satisfies the graded Leibniz rules against \eqref{bracs}.
\end{lemma}
\begin{proof}
Consider the non-central terms first. We have
    \begin{align*}
        [\hat d\bar\chi,\eta] = ([-t\bar x,y],[d\bar x,y]+[-t\bar x,Y],[-t\bar x,\mathsf{Y}] + [d\bar X-t\bar{\mathsf{X}},y]).
    \end{align*}
        On the other hand we have
    \begin{align*}
        [\bar\chi,\hat d\eta] =(0,[\bar x,dy-tY],[\bar X,dy-tY]),
    \end{align*}
    hence adding them gives, from \eqref{evencoad2},
    \begin{align*}
        [\hat d\bar\chi,\eta]+(-1)^{|\bar\chi|}[\bar\chi,\hat d\eta] &= \Big([-t\bar x,y],[d\bar x,y]+[\bar x,dy] -([t\bar x,Y]+ [\bar x,tY]),\\
        &\qquad\qquad [d\bar X,y] - [\bar X,dy] - ([t\bar{\mathsf{X}},y]+ [\bar{\mathsf{X}},ty])\Big) \\ 
        &= \Big(-t[\bar x,y], d[\bar x,y] - 0, d[\bar X,y]- 0\Big) \\
        &= \hat d\operatorname{ad}^*_{\bar\chi}(\eta),
    \end{align*}
    where we have used the product rule, the equivariance of $t$ and the Peiffer identity $[t\bar x,Y]+[\bar x,tY]=0$ \eqref{pfeif}. Similarly, we have
    \begin{equation*}
        [\hat d\chi,\bar\eta]+(-1)^{|\chi|}[\chi,\hat d\bar\eta] =  \hat d\operatorname{ad}^*_{\chi}(\bar\eta)
    \end{equation*}
    from \eqref{evencoad1}.

    Now let us turn to the central pieces $s^\lambda_+,s_-^\lambda$ in \eqref{central}. We first have
    \begin{equation*}
        s_+^\lambda(\hat d\bar\chi,\eta) =\lambda_1 \int_\Sigma\langle d\bar x,dy\rangle + \langle-t\bar x,dY\rangle  -\lambda_0\int_\Sigma\langle d\bar x,tY\rangle,
    \end{equation*}
    while
    \begin{equation*}
        s_-^\lambda(\bar\chi,\hat d\eta) = \lambda_0\int_\Sigma \langle \bar x,d(dy-tY)\rangle - \lambda_1\int_\Sigma\langle \bar x,tdY\rangle\,.
    \end{equation*}
    Provided $\Sigma$ has no boundary, $\langle d\bar x,dy\rangle = d\langle \bar x,dy\rangle$ is exact and hence vanishes under the integral by Stokes's theorem. We thus have
    \begin{align*}
        s_+^\lambda(\hat d\bar\chi,\eta)+ (-1)^{|\bar\chi|}s_-^\lambda(\bar\chi,\hat d\eta) &= -\lambda_0\int_\Sigma \langle t\bar x,dY\rangle - \langle \bar x,tdY\rangle \\ 
        &\qquad\qquad -\lambda_1\int_\Sigma\langle d\bar x,tY\rangle + \langle \bar x,d(tY)\rangle\\
        &= -\lambda_1\int_\Sigma d\langle \bar x,tY\rangle =0,
    \end{align*}
    where the first term vanishes by the equivariance $\langle t-,-\rangle = \langle -,t-\rangle$.
\end{proof}

Now the graded Jacobi identity for the non-central pieces is clear from \eqref{bracs}. For the central terms, this is equivalent to the \textit{graded 2-cocycle condition} for $s_+^\lambda, s_-^\lambda$, which one can directly compute to verify. The even/odd subcomplexes \begin{equation}
    \widehat{\Sigma_s\g}^+ = \Sigma\g^+\oplus s_+\mathbb{C}^{1|1}[1],\quad \widehat{\Sigma_s\g}^- = \Sigma\g^-\oplus s_-\mathbb{C}^{1|1}[1]\subset\widehat{\Sigma_s\g}\label{coisotropic}
\end{equation} 
are by construction coisotropic under \eqref{2kmpairing}.

\end{proof}
\noindent The Kac-Moody 2-algebra $(\widehat{\Sigma_s\g})$ may therefore be understood as a dg Manin triple/\textbf{2-Drinfel'd double} \cite{Bai_2013,Alfonsi_2023,chen:2022}. In analogy with \eqref{doubleCYB}, the fact that the associated projectors in a 2-Drinfel'd double $\d=(\g_+,\g_-,\langle-,-\rangle)$ --- where $\langle-,-\rangle$ is a \textit{degree -1} non-degenerate coisotropic pairing --- give rise to a solution $R$ of the 2-graded Yang-Baxter equation was known in \cite{Bai_2013}.

It is interesting that, in contrast to the usual Kac-Moody case in \S \ref{wzw}, each of the coisotropic subspaces themselves hosts central extensions $s_+,s_-$. However, it is important to notice that, due to the parity of the brackets \eqref{bracs}, only the \textit{positive}/chiral coisotropy subcomplex $\widehat{\Sigma_s\g}^+$ forms a dg Lie subalgebra.

\section{2-Lax pairs and the affine Kac-Moody 2-algebra}\label{2lax2km}
As recalled in \S \ref{wzw}, the Lax equation 
valued in the affine Kac-Moody algebra  can be rewritten as a 2d flatness condition 
\begin{equation*}
    \partial_t A_x - \partial_x A_t + [A_x,A_t]=0
\end{equation*}
for a {\bf Lax connection} $A = A_tdt+A_xdx$, where $A_x\oplus 1= L$ and $A_t\oplus 0=P$. Moreover, we have seen that these Lax pairs manifest themselves naturally in the 2d Wess-Zumino-Witten model. 

The goal of this section is to demonstrate how the derived Kac-Moody current 2-algebra underpins the 2-graded Lax integrability formalism introduced in \S \ref{2integra}. We show that the coadjoint structure \eqref{evencoad1} of the Kac-Moody current 2-algebra allows us to recover the 2-Lax equation in terms of a zero 2-curvature condition. 

To emphasize how the \textit{derived}/higher-homotopy nature our formalism arises, we will explicitly work out an example of a 3d topological-holomorphic field theory which, like the 2d Wess-Zumino-Witten, is 2-Lax integrable. 

\subsection{Zero 2-curvature representation of the 2-Lax equation}\label{2laxzero}
Suppose the tuple $\lambda=(\lambda_0,\lambda_1) = (1,1)$ in the central extension cocycle $s=s^{\lambda}$ is unity. We begin by taking the following coisotropic elements in $\widehat{\Sigma_s\g}$,
\begin{equation}
    P=(U_0,V_1,0)\oplus 0\in \Sigma\g^+ \oplus \mathbb{C}^{1|1}[1],\qquad L = (0,U_1,V_0)\oplus 1\in \Sigma\g^- \oplus \mathbb{C}^{1|1}[1],\label{laxpres}
\end{equation} 
for which
\begin{align*}
    U_0\in \Omega^0(\Sigma,\g_0),&\qquad V_0\in \Omega^2(\Sigma,\g_{-1}) \\
    V_1\in \Omega^1(\Sigma,\g_{-1}),&\qquad U_1\in\Omega^1(\Sigma,\g_0),
\end{align*}
with $L$ carrying a unit central element $\bar c=1\oplus 1\in\bbC^{1|1}$. Note $P,L$ have \textit{distinct parity} in their grading, in parallel with the degree splitting of the canonical 2-Lax pair \eqref{2lax}. 

Now on the 3-manifold $M= \Sigma\times \mathbb{R}$, we pick an integral curve $\gamma:[0,1]\to M$  generated by the timelike vector $\partial_\tau$ along $\mathbb{R}$ direction as in \S \ref{wzw}.\footnote{In this particular case, we shall simplify our notation $\partial_t\gamma^*L$ to $\dot L$.} The 2-Lax equation \eqref{2grlax} then can be computed from \eqref{evencoad1},
\begin{align*}
    \dot{L} &= (0,\partial_\tau U_1,\partial_\tau V_0) = \operatorname{ad}^*_PL \\
    &=\big(0,[U_0,U_1]+ dU_0- tV_1,[U_0,V_0] + [V_1,U_1] - dV_1\big)
\end{align*}
which on a local complex chart $(z,\bar z)$ of $\Sigma$, we can rename these as
\begin{align*}
    & U_1 = A_zdz + A_{\bar z}d\bar z, \qquad V_0 = B_{z\bar z}dz\wedge d\bar z, \\
    & U_0 = A_\tau, \qquad V_1 = B_{z\tau }dz  + B_{\bar z\tau }d\bar z
\end{align*}
to deduce the following equations
\begin{align*}
    \partial_\tau A_z &- \partial_z A_\tau + [A_\tau,A_z] = tB_{z\tau}, \\
    \partial_\tau A_{\bar z} &- \partial_{\bar z} A_\tau + [A_\tau,A_{\bar z}] = tB_{\bar z\tau},\\
    \partial_\tau B_{z\bar z} &- \partial_{\bar z} B_{z\tau } + \partial_zB_{\bar z\tau } \nonumber\\
    &\qquad + [A_\tau,B_{z\bar z}] - [A_{\bar z},B_{z\tau }] + [A_z,B_{\bar z\tau}] = 0.
\end{align*}
Aside from the singular leftover fake-flatness equation on the restriction $\iota_\Sigma^* (A,B)$ on $\Sigma\xhookrightarrow{\iota_\Sigma} \Sigma\times\R$,
\begin{equation}
    \partial_zA_{\bar z} - \partial_{\bar z}A_z + [A_z,A_{\bar z}] = tB_{z\bar z}, \label{2dfakeflat}
\end{equation}
these constitute nothing but the fake- and 2-flatness conditions 
\begin{equation*}
    dA + \frac{1}{2}[A,A] = tB,\qquad dB +[ A, B ]= 0
\end{equation*}
for a 2-connection $(A,B)$ on $\Sigma\times \mathbb{R}$, where 
\begin{align*}
    A &= A_\tau d\tau + A_zdz +A_{\bar z}d\bar z \in \Omega^1(\Sigma\times\mathbb{R})\otimes\g_0\\
    B&= B_{z\bar z}dz \wedge d\bar z + B_{z\tau}dz\wedge d\tau  + B_{\bar z \tau }d\bar z\wedge d\tau \in \Omega^2(\Sigma\times\mathbb{R})\otimes\g_{-1}.
\end{align*}

\begin{theorem}
    In the case of the affine Kac-Moody 2-algebra $\widehat{\Sigma_s\g}$, the 2-Lax equations \eqref{2grlax} along a timelike integral curve can be rewritten as the flatness conditions of a 2-connection $(A,B)$ on $\Sigma\times\mathbb{R}$ (aside from \eqref{2dfakeflat}).
\end{theorem}

Similar to the ordinary Lax equation, the 2-Lax equation has an inherent notion of "time" $\tau$ that runs normal to the Riemann surface $\Sigma$. More generally, such a "notion of time" on 3-manifolds $M^3$ is formalized as a {\it transverse holomorphic foliation} (THF) \cite{Aganagic:2017tvx,Jurco:2018sby} with leaves along $d\tau$. 
\begin{definition}
    A \textbf{transverse holomorphic foliation} on a 3-manifold $M^3$ is a smooth atlas with local charts of the form $(\tau,z,\bar z)\in \R\times\bbC$ on $M^3$, such that coordinate transitions take the form
\begin{equation*}
    (\tau,z,\bar z)\mapsto (\tau'(\tau,z,\bar z),z'(z),\bar z'(\bar z)).
\end{equation*}
$(z,\bar z)$ can be treated as local complex coordinates on the "spatial" slice $\Sigma\subset M^3$. 
\end{definition}
The product $M^3 = \R\times\Sigma$ is an example of a 3-manifold with THF structure. More generally, the THF structures are characterized by resolutions $\mathcal{A}=R\Gamma(\mathbb{R}\text{av})$ of the global sections on the \textit{raviolo} $\mathbb{R}\text{av}$ \cite{Kamnitzer:2022,Garner:2023zqn}, for which an analogous version of the affine Kac-Moody 2-algebra $\widehat{\mathcal{A}_s\g}$ can be constructed. 

\begin{remark}
    More precisely, the 2-Lax equation taking values in $\widehat{\Sigma_s\g}$ gives rise to 2-connections $(A,B)$ which are 2-flat (ie. zero 3-form 2-curvature $d_AB=0$), but is allowed to be non-fake-flat $(F_A-tB)\mid_\Sigma\neq0$ along the spatial slice $\Sigma$. This is because the 2-Lax equation \eqref{2lax} only governs the differential constraints involving "time", and therefore does not capture those higher flatness conditions that are parallel to $\Sigma$. 
\end{remark}

The same computations as above would then cast the 2-Lax equations \eqref{2grlax}, where the Lax pair $(L,P)$ may be \textit{split} Dolbeault-de Rham forms (see \S \ref{splitdeRham}), as certain topological-holomorphic differential constraints on the tuple $(A,B)$.

\medskip

The above observation is a higher homotopy analogue of ordinary Lax integrability (see \S \ref{wzw}): it allows one to cast 2-Lax pairs as defined in \eqref{2grlax} as 2-flat $G$-connections \cite{Baez:2004in,Wockel2008Principal2A}. In the following section, we will explicitly demonstrate the appearance of such a 2-Lax pair from a 3d field theory studied in \cite{Chen:2024axr}.

\subsection{2-Lax formulation of a 3d topological-holomorphic field theory}\label{2laxintheift}
Given a {balanced} Lie 2-algebra $(\g=\g_{-1}\xrightarrow{t}\g_0,\langle-,-\rangle)$ and its corresponding Lie 2-group $G=G_{-1}\xrightarrow{\textbf{t}} G_0$, one of the authors has constructed in \cite{Chen:2024axr} a 3d field theory $\mathcal{W}$ on a foliated 3d manifold $M^3$. The goal of this section is to show how this field theory fits in our proposed 2-lax formulation. 

\begin{remark}\label{integrableemphasis}
    It is worth noting here that the 3d field theory $\mathcal{W}$ was obtained from a higher-dimensional generalization (see \cite{Chen:2024axr,Schenkel:2024dcd}) of the 4d-2d localization of Costello-Yamazaki \cite{Costello:2019tri}. As such, following from its vvery construction, its integrability is expected but not yet formalized. This paper serves as the starting point for such a higher-dimensional integrability framework.
\end{remark}

To set up the action $\mathcal{W}$, consider an oriented 3-manifold $M^3$ equipped with a THF. For the purposes of dynamical analysis, we will assume $M^3=\Sigma\times \R$ splits along the THF (cf. pg. 7 of \cite{Garner:2023zqn}), where $\Sigma$ is a compact Riemann surface. We then have the following decomposition
\begin{equation*}
    \Omega^\bullet(M^3)\cong \Omega^\bullet(\Sigma)\otimes \Omega^\bullet(\R),\qquad d = d'+ d_\tau
\end{equation*}
of the de Rham complex. Note here that the differential also decomposes, for which $d'=\partial+\bar\partial$ is the exterior derivative on $\Sigma$. 

Locally, we consider a \textbf{chirality vector field} which locally takes the form $\partial_\ell$, which can point towards either one of the THF coordinates $\ell\in\{z,z',\tau\}$. Let $d_\ell$ denote the corresponding differential, then the action is given by
\begin{equation}
\mathcal{W}[g,\Theta\mid\ell] = -\int_{M^3} \langle d_\ell (d gg^{-1}), \Theta^g\rangle - \frac{1}{2}\langle d_\ell \Theta^g , t(\Theta^g)\rangle,\label{covariantaction}
\end{equation}
which takes as input the fields $(g,\Theta)\in C^\infty(M^3)\otimes G_0\oplus \Omega^1(M^3)\otimes \g_{-1}$ valued in the (derived \cite{Zucchini:2021bnn}) 2-group $G$ (here $\Theta^g= g\rhd \Theta$). Notice that the component $\Theta_\ell$ of the field $\Theta$ along the chirality $\ell$ is absent!

The equations of motion (EOMs) of the theory are equivalent to the fake- and 2-flatness condition
\begin{align}
    d_\ell \Lambda &=0,\qquad d_\ell H =0,\nonumber\\
    d_\perp \bar\Lambda &= t\bar H,\qquad d_\perp \bar H =0\label{eoms}
\end{align}
living in $\Omega^1(M^3) \otimes\g_0 \oplus \Omega^2(M^3)\otimes\g_{-1}$, for the higher-form currents
\begin{align}
    j& = (\Lambda, H)  = \left(-\big(d_\perp gg^{-1}+t\Theta^g\big)\,,\, g\rhd (d_\perp \Theta - \frac{1}{2}[\Theta,\Theta])\right) \nonumber\\ 
   \bar\jmath &= (\bar\Lambda, \bar H)  =  \left(g^{-1}d_\ell g\,,\, -\big(d_\ell \Theta + [g^{-1}d_\ell g,\Theta]\big)\right)\label{highercurrents}
\end{align}
defined in terms of the fields $(g,\Theta)$ of the theory.

\subsubsection{Split de Rham complexes}\label{splitdeRham}
From here on, we will specify for simplicity that the chirality vector field in fact aligns along the leaves of the THF --- namely, $\ell=\tau$. Define the following split complexes
\begin{align}
    \mathcal{B}^{\bullet,0} = \Omega^\bullet(\Sigma)\otimes \Omega^0(\R), \qquad \mathcal{B}^{\bullet,1} = \Omega^\bullet(\Sigma)\otimes \Omega^1(\R),\label{splitcomplexes}
\end{align}
each of which comes equipped with the differential $d_\perp =d$ given by the exterior derivative on $\Sigma$,
\begin{equation*}
    d_\perp: \mathcal{B}^{\bullet,q}\rightarrow \mathcal{B}^{\bullet+1,q},\qquad q=0,1.
\end{equation*}
Notice that, by adding back in the derivative $d_\tau$ along the $\R$ direction, $$\mathcal{B}^{\bullet,q}=\Big((\Omega^\bullet(\Sigma)\otimes \Omega^q(\R),d_\perp),d_\tau\Big)\cong \Omega^\bullet(\Sigma\times\R)\,,$$ we recover the full de Rham complex on $\R\subset M^3=\Sigma\times\R$.

With this choice, it is then clear that the higher-currents $j,\bar\jmath$ \eqref{highercurrents} split in accordance with the complexes $\mathcal{B}^{\bullet,q}$
\begin{align*}
    &\Lambda\in\mathcal{B}^{1,0}\otimes\g_0,\quad H\in\mathcal{B}^{2,0}\otimes\g_{-1},\\
    &\bar\Lambda \in \mathcal{B}^{0,1}\otimes\g_0,\quad \bar H\in \mathcal{B}^{1,1}\otimes\g_{-1},
\end{align*}
or that, equivalently, these higher-currents live 
\begin{equation*}
    j=(\Lambda,H) \in \big(\mathcal{B}^{\bullet,0}\otimes\g\big)_1,\qquad \bar\jmath=(\bar\Lambda,\bar H) \in \big(\mathcal{B}^{\bullet,1}\otimes\g\big)_1
\end{equation*}
in the complexes $\mathcal{B}^{\bullet,q}\otimes\g$ with total degree one. This formulation drastically simplifies their EOMs \eqref{eoms},
\begin{equation*}
    d_\ell j =0,\qquad \hat d_\perp \bar\jmath = 0,
\end{equation*}
where $\hat d_\perp=d_\perp-t$ is the differential on $\mathcal{B}^{\bullet,1}\otimes\g$.

\begin{remark}\label{raviolo}
    We emphasize that the chirality $\partial_\ell$ is in general independent from the THF on $M^3$. We have aligned them $\partial_\ell=\partial_\tau$ here to fit with the example construction of \S \ref{zero2curv}, but it is possible to misalign them. If we make the holomorphic choice $\partial_\ell= \partial_z$, for instance, then the higher currents $j,\bar\jmath$ \eqref{highercurrents} of the theory $\mathcal{W}$ would live in a \textit{differently-split} dgca $\mathcal{A}^{\bullet,q}$, where $q=0,1$ now counts the  $dz$ legs. These mixed Dolbeault-de Rham complexes $\mathcal{A}^{\bullet,q}$ turn out to resolve global sections on the \textit{raviolo} $\mathbb{R}\operatorname{av}$ \cite{Garner:2023zqn,alfonsi2024raviolo}, and the corresponding Kac-Moody 2-algebra $\widehat{\mathcal{A}^{\bullet,q}_s\g}$ (or more precisely its cohomology $H(\widehat{\mathcal{A}^{\bullet,q}_s\g})$) is studied in \cite{Chen:2025ujx}.
\end{remark}

By specifying the orientation of $M^3=\Sigma\times \R$, the 3d Hodge star satisfies
\begin{equation*}
    \operatorname{vol}_\Sigma = \star_3 d\tau
\end{equation*}
with respect to the volume form $\operatorname{vol}_\Sigma$ on $\Sigma$. This then defines a local linear isomorphism
\begin{equation*}
    \star_3: \mathcal{B}^{p,1}\rightarrow \mathcal{B}^{2-p,0}\,,\qquad p=1,2
\end{equation*}
which allows us to view the higher-currents $\bar\jmath = (\bar\Lambda,\bar H)$ as forms on $\Sigma$,
\begin{equation*}
    \star_3\bar\Lambda \in\Omega^2(\Sigma,\g_0),\qquad \star_3\bar H\in\Omega^1(\Sigma,\g_{-1}).
\end{equation*}
In the following, we will work exclusively on $\Sigma$ and denote by $\tilde\Lambda=\star_3\bar\Lambda\,,\tilde H=\star_3\bar H$.

\subsubsection{2-Lax pairs corresponding to the 2-flat currents} Recall the chirality is aligned with the THF, $\partial_\ell=\partial_\tau$. Through the 3d Hodge star operator $\star_3$ on $M^3=\Sigma\times \R$, we have 
\begin{align*}
    &\Lambda\in \Omega^1(\Sigma,\g_0),\qquad H\in\Omega^2(\Sigma,\g_{-1}),\\
    &\tilde\Lambda\in\Omega^2(\Sigma,\g_0),\qquad \tilde H\in\Omega^1(\Sigma,\g_{-1})\,.
\end{align*}
Notice that these are precisely the graded components of the current dg algebra $\Omega^\bullet(\Sigma,\g)$ defined in \S \ref{zero2curv}, with $j \in \Sigma\g^-$ in the odd subsector and $\bar\jmath\in\Sigma\g$ in the even one.

\medskip

From the 2-flatness equations of motion for $(j,\bar\jmath)$, the correspondence described in \S \ref{2laxzero} then suggests the following 2-Lax pairs
\begin{align}
    &\bar L =  (0,\Lambda,H)\oplus1,\qquad \bar P=0\oplus 0  & \in \Sigma\g^-\oplus s_-\mathbb{C}^{1|1}[1]\nonumber\\
    &  L = 0\oplus 1,\qquad  P = (0,\tilde H,\tilde\Lambda)\oplus 0 &\in \Sigma\g^+\oplus s_+\mathbb{C}^{1|1}[1]
    ,\label{ift2lax}
\end{align}
as written in the form of \eqref{laxpres}. Indeed, consistent with \S \ref{2laxzero}, the following two 2-Lax equations
\begin{equation}
    \partial_\tau \bar L = 0,\qquad 0=\operatorname{ad}_{P}^*(0\oplus 1) \label{ift2laxeq}
\end{equation}
can be seen to be equivalent to the EOMs \eqref{eoms} of the theory $\mathcal{W}$. 

There is then a direct correspondence between solutions of \eqref{eoms} with 2-Lax pairs of the form \eqref{ift2lax} valued in the coisotropic subcomplexes \eqref{coisotropic} of the Kac-Moody 2-algebra $\widehat{\Sigma_s\g}$. This is in direct analogy with the 2d Wess-Zumino-Witten model case described in \S \ref{wzw}.

\medskip

Turning to the  symmetries of the theory itself we will now prove (\textbf{Theorem \ref{2kmsymm}}) that the global symmetries of the theory $\cal W$ \eqref{covariantaction} are governed by the Kac-Moody current 2-algebra, analogous to the global Kac-Moody symmetry of the 2d Wess-Zumino-Witten model $W$ (see \cite{KNIZHNIK198483} and \S \ref{wzw}).

\subsection{Higher Kac-Moody symmetry charges}\label{momentmap}
We begin by briefly recalling the Noether analysis of $\cal W$ performed in \cite{Chen:2024axr}. The global left- and right-acting infinitesimal symmetries are given by
\begin{equation*}
    \delta_{(\alpha,\Gamma)} (g,\Theta)= (g.\alpha,-[\alpha,\Theta]+\Gamma),\qquad \delta_{(\bar \alpha,\bar\Gamma)}(g,\Theta) = (-\bar \alpha.g,-g\rhd \bar\Gamma),
\end{equation*}
where the symmetry parameters
\begin{align*}
   (\alpha,\Gamma),(\bar\alpha,\bar\Gamma) \in \big(\mathcal{B}^{0,0}\otimes\g_0\big)\oplus \big(\mathcal{B}^{1,0}\otimes\g_{-1}\big) = \big(\mathcal{B}^{\bullet,0}\otimes\g\big)_0
\end{align*}
must satisfy
\begin{align}
    & d_\tau\alpha =0,\qquad d_\tau \Gamma =0 \nonumber \\ 
    & d_\perp\bar\alpha + t\bar\Gamma = 0,\qquad d_\perp\bar\Gamma = 0.\label{2gtbcs}
\end{align}
Put the local charge aspects of the theory (they are \textit{functions on field space} in the covariant phase space formalism \cite{Geiller:2020edh})
\begin{equation*}
    Q_{(\alpha,\Gamma)} = \langle \alpha,H\rangle + \langle \Gamma,\Lambda\rangle\in\cB^{2,0} ,\qquad \bar Q_{(\bar\alpha,\bar\Gamma)}=\langle \bar\alpha,\bar H\rangle + \langle\bar\Gamma,\bar \Lambda\rangle\in\cB^{1,1},
\end{equation*}
it was shown in \cite{Chen:2024axr} that ---  on-shell of the EOM \eqref{eoms} and the constraints \eqref{2gtbcs} --- they are conserved $dQ=d\bar Q=0$ and satisfy the following graded Poisson bracket relatinos
\begin{align}
    [Q_{(\alpha,\Gamma)},\bar Q_{(\bar \alpha,\bar\Gamma)}] &=0\nonumber\\
    [Q_{(\alpha,\Gamma)},Q_{(\alpha',\Gamma')}] &= -Q_{([\alpha',\alpha],[\alpha',\Gamma]+ [\Gamma',\alpha])} +\langle \Gamma,d_\perp\alpha'+t\Gamma'\rangle + \langle \alpha,d_\perp\Gamma'\rangle \label{noether1} \\
     [\bar Q_{(\bar\alpha,\bar\Gamma)},\bar Q_{(\bar\alpha',\bar\Gamma')}]&= -\bar Q_{([\bar \alpha',\bar \alpha],[\bar \alpha',\bar \Gamma]+ [\bar \Gamma',\bar \alpha])} - \langle \bar \Gamma,d_\perp\bar \alpha'\rangle - \langle \bar \alpha,d_\perp\bar \Gamma'\rangle.\label{noether2}
\end{align}
In the following, we will focus on the \textit{chiral} sector $Q$ of the theory $\mathcal{W}$, corresopnding to the left-acting global symmetries $(\alpha,\Gamma).$ This sector is the 3d analogue of the \textit{holomorphic} left-acting symmetries in the Wess-Zumino-Witten model.

Denote by $\D,{\bar\D}\subset (\mathcal{B}^{\bullet,0}\otimes\g)_0$ the space of solutions to \eqref{2gtbcs} with graded Lie algebra structures \eqref{noether1}, \eqref{noether2}. 
\begin{theorem}\label{2kmsymm}
    Suppose $\partial_\ell=\partial_\tau$ for the theory $\cal W$, and let $\D$ denote the space of symmetry charges $Q$ satisfying \eqref{noether1}, then there is a dg Lie algebra map $$\D \rightarrow \widehat{\Sigma_s\g}^+$$ identifying a global 2-Kac-Moody symmetry of $\cal W$.
\end{theorem}
\begin{proof}
The requisite map is defined by
\begin{equation*}
    (\alpha,\Gamma)\mapsto (\alpha,\Gamma,0)\in\Sigma\g^+.
\end{equation*}
This map is evidently dg linear, and we seek to recover the graded Lie algebra structure \eqref{noether1} of the charges from the structure of the positive coisotropy complex $\widehat{\Sigma_s\g}^+$ of the Kac-Moody 2-algebra.

Recall from \eqref{ift2lax} that $(0,\Lambda,H)\in \Sigma\g^-$. Notice that, upon integrating over $\Sigma$ the Noether charge is nothing but the contraction with \eqref{2kmpairing},
\begin{equation}
    \int_\Sigma Q_{(\alpha,\Gamma)} = \llparenthesis (\alpha,\Gamma,0)\,|\, (0,\Lambda,H)\rrparenthesis.
    \label{chargecontraction}
\end{equation}
Putting $\chi=(\alpha',\Gamma',0),\,\eta=(\alpha,\Gamma,0)$
, we have from the dg Lie brackets \eqref{bracs}
\begin{align*}
     \llparenthesis [\chi,\eta]\,|\, (0,\Lambda,H)\oplus 1\rrparenthesis &= \llparenthesis [([\alpha',\alpha],[\Gamma',\alpha]+[\alpha',\Gamma],0)]\,|\, (0,\Lambda,H)\oplus 1\rrparenthesis\\
     &= \int_\Sigma \langle [\alpha',\alpha],H\rangle + \langle [\Gamma',\alpha]+ [\alpha',\Gamma],\Lambda\rangle + \langle \alpha',d\Gamma\rangle + \langle \Gamma',d\alpha\rangle - \langle \Gamma',t\Gamma\rangle \\
     &\overset{\eqref{chargecontraction}}{=} \int_\Sigma Q_{[(\alpha',\Gamma'),(\alpha,\Gamma)]}- \int_\Sigma \langle \alpha,d\Gamma'\rangle + \langle \Gamma,d\alpha'+ t\Gamma'\rangle \\ 
\end{align*}
which reproduces for us precisely the integrated version of the graded charge algebra \eqref{noether1}. Whence, under the pairing form \eqref{2kmpairing}, we indeed have
\begin{equation*}
    [\chi,\eta] = -[Q_{(\alpha,\Gamma)}, Q_{(\alpha,\Gamma)}]
\end{equation*}
as desired.

\end{proof}
The above result proves the statement that was expected in the previous work \cite{Chen:2024axr}: the chiral charges $Q$ of the theory $\mathcal{W}$ give rise to (infinitely many) conserved higher-dimensional monodromy operators labeled by dg representations of the subalgebra $\widehat{\Sigma_s\g}^+$.

\medskip

Here we briefly make some concluding comments. The \textit{shifted tangent complex} $M = T[1]\Sigma$ can be understood as a dg manifold with the sheaf $\mathcal{O}_M = \Omega^\bullet(\Sigma)$ given by the de Rham cochain complex on $\Sigma$. As such, the Kac-Moody 2-algebra $\widehat{\Sigma_s\g}$ can be thought of as a dg central extension of the $\g$-valued functions $\Sigma\g = \Omega^\bullet(\Sigma,\g) = C(M)\otimes\g$ on $M=T[1]\Sigma$. 

Through this perspective, one can then play the same game with a different choice of the dg manifold $(M,\mathcal{O}_M)$, which leads to {different} Kac-Moody 2-algebras $\widehat{\mathcal{A}_s\g}$, where $\mathcal{A}=\mathcal{O}_M$ is the given dgca (cf. the beginning of \S \ref{zero2curv}). Indeed, choosing the \textit{raviolo} $M=\mathbb{R}\operatorname{av}$ leads to the raviolo vertex algebra \cite{Garner:2023zqn,alfonsi2024raviolo,Chen:2025ujx}, while the choice  $\mathcal{O}_M=R\Gamma_{hol}(D^\times)$ (the dgca resolving the sheaf of holomorphic sections on the formal punctured $n$-disc $D^\times\subset \mathbb{C}^n$) would relate to the derived current algebras of \cite{FAONTE2019389,Alfonsi_2023}.

\section{Conclusion}
The notion of 2-graded integrability, inherited from an underlying Poisson 2-algebra, can be interpreted as a categorification (in the sense described in the introduction) of the usual notion of integrability for 1d systems to 2-dimensions. This paper sought to make this understanding more precise, with the construction of the Kac-Moody 2-algebra in \S \ref{zero2curv}, as well as an application of its structures to a 3-dimensional field theory which hosts 2-flat higher Lax connections.

By generalizing Lax integrability from Lie 1-bialgebras \cite{Lax:1968fm} (and following \cite{Meusburger:2021cxe}), an explicit formula \eqref{2lax} for the dg Lax operators $L,P$ from a solution of the 2-graded CYBE \cite{Bai_2013,Chen:2012gz,chen:2022} was given. Importantly, it was found that $L,P$ have an \textit{even-odd splitting} in the underlying degrees, which turned out to be crucial for the reformulation of 2-Lax equations in terms of 2-flatness conditions in \S \ref{2laxzero}.

The authors expect many aspects of our work to be related to previous works on derived generalizations of the Kac-Moody algebra in the literature \cite{Kapranov2021InfinitedimensionalL,FAONTE2019389}. Given the close relation between the Kac-Moody 2-algebra $\widehat{\Sigma_s\g}$ and higher Lax pairs, our work suggests an interesting prospect of understanding higher-dimensional field theories through the lens of dg integrability in the context of derived geometry.

\medskip

As also mentioned briefly in \S \ref{2grpois2alg}, it would also be interesting to generalize our higher Lax integrability to the case of $L_\infty$-algebras and quasi-Lie 2-bialgebras \cite{Chen:2013,Chen:2012gz}. Specifically, applications to the weak string 2-algebra \cite{Schommer_Pries_2011,Waldorf2012ACO,Baez:2005sn,Cirio:2012be} would be desirable.

\newpage

\printbibliography

@INPROCEEDINGS{Voronov:2008,
       author = {{Khudaverdian}, H.~M. and {Voronov}, Th. Th.},
        title = "{Higher Poisson Brackets and Differential Forms}",
    booktitle = {Geometric Methods in Physics},
         year = 2008,
       editor = {{Kielanowski}, Piotr and {Odzijewicz Anatol} and {Schlichenmaier}, Martin and {Voronov}, Theodore},
       series = {American Institute of Physics Conference Series},
       volume = {1079},
        month = nov,
    publisher = {AIP},
        pages = {203-215},
          doi = {10.1063/1.3043861},
archivePrefix = {arXiv},
       eprint = {0808.3406},
 primaryClass = {math-ph},
       adsurl = {https://ui.adsabs.harvard.edu/abs/2008AIPC.1079..203K}
}

@article{SAFRONOV2021107633,
title = {Poisson-Lie structures as shifted Poisson structures},
journal = {Advances in Mathematics},
volume = {381},
pages = {107633},
year = {2021},
issn = {0001-8708},
doi = {https://doi.org/10.1016/j.aim.2021.107633},
url = {https://www.sciencedirect.com/science/article/pii/S0001870821000712},
author = {Pavel Safronov}
}

@ARTICLE{Seol2022-sr,
  title    = "Dg Manifolds, Formal Exponential Maps and Homotopy Lie Algebras",
  author   = "Seol, Seokbong and Sti{\'e}non, Mathieu and Xu, Ping",
  journal  = "Communications in Mathematical Physics",
  volume   =  391,
  number   =  1,
  pages    = "33--76",
  month    =  apr,
  year     =  2022
}

@misc{Kontsevich_deformation,
    author = "Kontsevich, Maxim and Soibelman, Yan",
    title = "{Deformation Theory I}",
    journal = "Unpublished book draft",
    url="https://people.maths.ox.ac.uk/beem/papers/kontsevich_soibelman_deformation_theory_1.pdf",
    year = "1990"
}

@article{Kontsevich:2006jb,
    author = "Kontsevich, Maxim and Soibelman, Yan",
    title = "{Notes on A$\infty$-Algebras, A$\infty$-Categories and Non-Commutative Geometry}",
    eprint = "math/0606241",
    archivePrefix = "arXiv",
    doi = "10.1007/978-3-540-68030-7_6",
    journal = "Lect. Notes Phys.",
    volume = "757",
    pages = "153--220",
    year = "2009"
}

@article{Alexandrov:1995kv,
    author = "Alexandrov, M. and Schwarz, A. and Zaboronsky, O. and Kontsevich, M.",
    title = "{The Geometry of the master equation and topological quantum field theory}",
    eprint = "hep-th/9502010",
    archivePrefix = "arXiv",
    reportNumber = "UCD-94-01, UCD{\&}B-94-01",
    doi = "10.1142/S0217751X97001031",
    journal = "Int. J. Mod. Phys. A",
    volume = "12",
    pages = "1405--1429",
    year = "1997"
}

@article{Calaque2015ShiftedPS,
  title={Shifted Poisson structures and deformation quantization},
  author={Damien Calaque and Tony Pantev and Bertrand To{\"e}n and Michel Vaqui{\'e} and Gabriele Vezzosi},
  journal={Journal of Topology},
  year={2015},
  volume={10},
  url={https://api.semanticscholar.org/CorpusID:117757610}
}

@article{pridham2018outline,
  title={An outline of shifted Poisson structures and deformation quantisation in derived differential geometry},
  author={Pridham, Jonathan P},
  journal={arXiv preprint arXiv:1804.07622},
  year={2018}
}

@article{pantev2013shifted,
  title={Shifted symplectic structures},
  author={Pantev, Tony and To{\"e}n, Bertrand and Vaqui{\'e}, Michel and Vezzosi, Gabriele},
  journal={Publications math{\'e}matiques de l'IH{\'E}S},
  volume={117},
  pages={271--328},
  year={2013}
}

@article{Kemp_2025,
   title={Infinitesimal 2-braidings from 2-shifted Poisson structures},
   volume={212},
   ISSN={0393-0440},
   url={http://dx.doi.org/10.1016/j.geomphys.2025.105456},
   DOI={10.1016/j.geomphys.2025.105456},
   journal={Journal of Geometry and Physics},
   publisher={Elsevier BV},
   author={Kemp, Cameron and Laugwitz, Robert and Schenkel, Alexander},
   year={2025},
   month=jun, pages={105456} }

@article{Kamnitzer:2022,
author = {Kamnitzer, Joel},
title = {Symplectic resolutions, symplectic duality, and Coulomb branches},
journal = {Bulletin of the London Mathematical Society},
volume = {54},
number = {5},
pages = {1515-1551},
doi = {https://doi.org/10.1112/blms.12711},
url = {https://londmathsoc.onlinelibrary.wiley.com/doi/abs/10.1112/blms.12711},
eprint = {https://londmathsoc.onlinelibrary.wiley.com/doi/pdf/10.1112/blms.12711},
year = {2022}
}

@article{Chen:2025ujx,
    author = "Chen, Hank and Liniado, Joaquin",
    title = "{Infinite Dimensional Topological-Holomorphic Symmetry in Three-Dimensions}",
    eprint = "2507.01858",
    archivePrefix = "arXiv",
    primaryClass = "hep-th",
    month = "7",
    year = "2025"
}

@article{Kuniba_2023,
   title={New Solutions to the Tetrahedron Equation Associated with Quantized Six-Vertex Models},
   volume={401},
   ISSN={1432-0916},
   url={http://dx.doi.org/10.1007/s00220-023-04711-y},
   DOI={10.1007/s00220-023-04711-y},
   number={3},
   journal={Communications in Mathematical Physics},
   publisher={Springer Science and Business Media LLC},
   author={Kuniba, Atsuo and Matsuike, Shuichiro and Yoneyama, Akihito},
   year={2023},
   month=may, pages={3247–3276} }

@misc{heredia2016representations2groupsbaezcrans2vector,
      title={On the representations of 2-groups in {Baez-Crans} 2-vector spaces}, 
      author={Benjamín Alarcón Heredia and Josep Elgueta},
      year={2016},
      eprint={1607.04986},
      archivePrefix={arXiv},
      primaryClass={math.CT},
      url={https://arxiv.org/abs/1607.04986}, 
}

@ARTICLE{Zamolodchikov:1980,
       author = {{Zamolodchikov}, A.~B.},
        title = "{Tetrahedra equations and integrable systems in three-dimensional space}",
      journal = {Soviet Journal of Experimental and Theoretical Physics},
         year = 1980,
        month = aug,
       volume = {52},
        pages = {325},
       adsurl = {https://ui.adsabs.harvard.edu/abs/1980JETP...52..325Z}
}

@ARTICLE{Semenov-Tyan-Shanskii1983-ti,
  title   = "What is a classical r-matrix?",
  author  = "Semenov-Tyan-Shanskii, M A",
  journal = "Functional Analysis and Its Applications",
  volume  =  17,
  number  =  4,
  pages   = "259--272",
  month   =  oct,
  year    =  1983
}

@misc{Chen1:2025?,
      title={Combinatorial quantization of 4d 2-Chern-Simons theory I: the Hopf category of higher-graph states}, 
      author={Hank Chen},
      year={2025},
      eprint={2501.06486},
      archivePrefix={arXiv},
      primaryClass={math-ph},
      url={https://arxiv.org/abs/2501.06486}, 
}

@article{Waldorf2012ACO,
  title={A Construction of String 2-Group Models using a Transgression-Regression Technique},
  author={Konrad Waldorf},
  journal={Contemp. Math.},
volume={584},
pages={99-115},
  year={2012},
  url={https://api.semanticscholar.org/CorpusID:119267307}
}

@article{Schommer_Pries_2011,
   title={Central extensions of smooth 2–groups and a finite-dimensional string 2–group},
   volume={15},
   ISSN={1465-3060},
   url={http://dx.doi.org/10.2140/gt.2011.15.609},
   DOI={10.2140/gt.2011.15.609},
   number={2},
   journal={Geometry \& Topology},
   publisher={Mathematical Sciences Publishers},
   author={Schommer-Pries, Christopher-J},
   year={2011},
   month=may, pages={609–676} }

@article{Khovanov_2010,
   title={A categorification of quantum $\mathrm{sl}(n)$},
   volume={1},
   ISSN={1664-073X},
   url={http://dx.doi.org/10.4171/QT/1},
   DOI={10.4171/qt/1},
   number={1},
   journal={Quantum Topology},
   publisher={European Mathematical Society - EMS - Publishing House GmbH},
   author={Khovanov, Mikhail and Lauda, Aaron D.},
   year={2010},
   month=feb, pages={1–92} }

@article{Rouquier20082KacMoodyA,
  title={2-Kac-Moody algebras},
  author={Raphael Rouquier},
  journal={arXiv: Representation Theory},
  year={2008},
  url={https://api.semanticscholar.org/CorpusID:115155841}
}

@article{Alekseev:1992wn,
    author = "Alekseev, A. and Faddeev, L. D. and Semenov-Tian-Shansky, M.",
    title = "{Hidden quantum groups inside Kac-Moody algebra}",
    doi = "10.1007/BF02097628",
    journal = "Commun. Math. Phys.",
    volume = "149",
    pages = "335--345",
    year = "1992"
}

@article{Gwilliam:2018lpo,
    author = "Gwilliam, Owen and Williams, Brian R.",
    title = "{Higher Kac\textendash{}Moody algebras and symmetries of holomorphic field theories}",
    eprint = "1810.06534",
    archivePrefix = "arXiv",
    primaryClass = "math.QA",
    doi = "10.4310/ATMP.2021.v25.n1.a4",
    journal = "Adv. Theor. Math. Phys.",
    volume = "25",
    number = "1",
    pages = "129--239",
    year = "2021"
}

@article{Tolstoy2002FromQA,
  title={From Quantum Affine Kac-Moody Algebras to Drinfeldians and Yangians},
  author={Valeriy N. Tolstoy},
  journal={arXiv: Quantum Algebra},
  year={2002},
  url={https://api.semanticscholar.org/CorpusID:15933171}
}

@article{Cirio:2012be,
    author = "Cirio, Lucio Simone and Martins, Joao Faria",
    title = "{Categorifying the $sl(2,C)$ Knizhnik-Zamolodchikov Connection via an Infinitesimal 2-Yang-Baxter Operator in the String Lie-2-Algebra}",
    eprint = "1207.1132",
    archivePrefix = "arXiv",
    primaryClass = "hep-th",
journal = "Adv. Theor. Math. Phys.",
volume="21", 
number="1", 
paegs="147 - 229", 
year = "2017"
}

@article{FAONTE2019389,
title = {Higher Kac–Moody algebras and moduli spaces of G-bundles},
journal = {Advances in Mathematics},
volume = {346},
pages = {389-466},
year = {2019},
issn = {0001-8708},
doi = {https://doi.org/10.1016/j.aim.2019.01.040},
url = {https://www.sciencedirect.com/science/article/pii/S0001870819300763},
author = {Giovanni Faonte and Benjamin Hennion and Mikhail Kapranov},
}

@article{Kapranov2021InfinitedimensionalL,
  title={Infinite-dimensional (dg) Lie algebras and factorization algebras in algebraic geometry},
  author={Mikhail Kapranov},
  journal={Japanese Journal of Mathematics},
  year={2021},
  pages={1-32},
  url={https://api.semanticscholar.org/CorpusID:231392391}
}

@article{KAC197885,
title = {Infinite-dimensional algebras, Dedekind's $\eta$-function, classical möbius function and the very strange formula},
journal = {Advances in Mathematics},
volume = {30},
number = {2},
pages = {85-136},
year = {1978},
issn = {0001-8708},
doi = {https://doi.org/10.1016/0001-8708(78)90033-6},
url = {https://www.sciencedirect.com/science/article/pii/0001870878900336},
author = {V.G Kac}
}

@article{Hoare:2021dix,
    author = "Hoare, Ben",
    title = "{Integrable deformations of sigma models}",
    eprint = "2109.14284",
    archivePrefix = "arXiv",
    primaryClass = "hep-th",
    doi = "10.1088/1751-8121/ac4a1e",
    journal = "J. Phys. A",
    volume = "55",
    number = "9",
    pages = "093001",
    year = "2022"
}

@article{Mohammedi:2020qok,
    author = "Mohammedi, N.",
    title = "{Some integrable deformations of the Wess-Zumino-Witten model}",
    eprint = "2012.09753",
    archivePrefix = "arXiv",
    primaryClass = "hep-th",
    doi = "10.1103/PhysRevD.104.126028",
    journal = "Phys. Rev. D",
    volume = "104",
    number = "12",
    pages = "126028",
    year = "2021"
}

@article{WITTEN1983422,
title = {Global aspects of current algebra},
journal = {Nuclear Physics B},
volume = {223},
number = {2},
pages = {422-432},
year = {1983},
issn = {0550-3213},
doi = {https://doi.org/10.1016/0550-3213(83)90063-9},
url = {https://www.sciencedirect.com/science/article/pii/0550321383900639},
author = {Edward Witten}
}

@article{WESS197195,
title = {Consequences of anomalous ward identities},
journal = {Physics Letters B},
volume = {37},
number = {1},
pages = {95-97},
year = {1971},
issn = {0370-2693},
doi = {https://doi.org/10.1016/0370-2693(71)90582-X},
url = {https://www.sciencedirect.com/science/article/pii/037026937190582X},
author = {J. Wess and B. Zumino}
}

@article{Schenkel:2024dcd,
    author = "Schenkel, Alexander and Vicedo, Beno\textasciicircum{}it",
    title = "{5d 2-Chern-Simons Theory and 3d Integrable Field Theories}",
    eprint = "2405.08083",
    archivePrefix = "arXiv",
    primaryClass = "hep-th",
    doi = "10.1007/s00220-024-05170-9",
    journal = "Commun. Math. Phys.",
    volume = "405",
    number = "12",
    pages = "293",
    year = "2024"
}

@article{Chen:2024axr,
    author = "Chen, Hank and Liniado, Joaquin",
    title = "{Higher gauge theory and integrability}",
    eprint = "2405.18625",
    archivePrefix = "arXiv",
    primaryClass = "hep-th",
    doi = "10.1103/PhysRevD.110.086017",
    journal = "Phys. Rev. D",
    volume = "110",
    number = "8",
    pages = "086017",
    year = "2024"
}

@article{alfonsi2024raviolo,
  title={Raviolo vertex algebras, cochains and conformal blocks},
  author={Alfonsi, Luigi and Kim, Hyungrok and Young, Charles AS},
  journal={arXiv preprint arXiv:2401.11917},
  year={2024}
}

@article{KNIZHNIK198483,
title = {Current algebra and Wess-Zumino model in two dimensions},
journal = {Nuclear Physics B},
volume = {247},
number = {1},
pages = {83-103},
year = {1984},
issn = {0550-3213},
doi = {https://doi.org/10.1016/0550-3213(84)90374-2},
author = {V.G. Knizhnik and A.B. Zamolodchikov}
}

@article{Aganagic:2017tvx,
    author = "Aganagic, Mina and Costello, Kevin and McNamara, Jacob and Vafa, Cumrun",
    title = "{Topological Chern-Simons/Matter Theories}",
    eprint = "1706.09977",
    archivePrefix = "arXiv",
    primaryClass = "hep-th",
    month = "6",
    year = "2017"
}

@article{Costello:2019tri,
    author = "Costello, Kevin and Yamazaki, Masahito",
    title = "{Gauge Theory And Integrability {III}}",
    eprint = "1908.02289",
    archivePrefix = "arXiv",
    primaryClass = "hep-th",
    reportNumber = "IPMU19-0110",
    month = "8",
    year = "2019"
}

@article{Jurco:2018sby,
    author = {Jur\v{c}o, Branislav and Raspollini, Lorenzo and S\"amann, Christian and Wolf, Martin},
    title = "{$L_\infty$-Algebras of Classical Field Theories and the Batalin-Vilkovisky Formalism}",
    eprint = "1809.09899",
    archivePrefix = "arXiv",
    primaryClass = "hep-th",
    reportNumber = "EMPG-18-19, DMUS-MP-18/05",
    doi = "10.1002/prop.201900025",
    journal = "Fortsch. Phys.",
    volume = "67",
    number = "7",
    pages = "1900025",
    year = "2019"
}

@article{Zucchini:2021bnn,
    author = "Zucchini, Roberto",
    title = "{4-d Chern-Simons Theory: Higher Gauge Symmetry and Holographic Aspects}",
    eprint = "2101.10646",
    archivePrefix = "arXiv",
    primaryClass = "hep-th",
    reportNumber = "DIFA UNIBO 2021",
    doi = "10.1007/JHEP06(2021)025",
    journal = "JHEP",
    volume = "06",
    pages = "025",
    year = "2021"
}

@article{Alfonsi_2023,
	doi = {10.1016/j.geomphys.2023.104903},
  
	url = {https://doi.org/10.1016%2Fj.geomphys.2023.104903},
  
	year = 2023,
	month = {sep},
  
	publisher = {Elsevier {BV}
},
  
	volume = {191},
  
	pages = {104903},
  
	author = {Luigi Alfonsi and Charles Young},
  
	title = {Higher current algebras, homotopy Manin triples, and a rectilinear adelic complex},
  
	journal = {Journal of Geometry and Physics}
}

@article{Stasheff:1963,
 ISSN = {00029947},
 URL = {http://www.jstor.org/stable/1993608},
 author = {James Dillon Stasheff},
 journal = {Transactions of the American Mathematical Society},
 number = {2},
 pages = {275--292},
 publisher = {American Mathematical Society},
 title = {Homotopy Associativity of H-Spaces. I},
 urldate = {2023-07-29},
 volume = {108},
 year = {1963}
}

@article{Porst2008Strict2A,
  title={Strict 2-Groups are Crossed Modules},
  author={Sven-S. Porst},
  journal={arXiv: Category Theory},
  year={2008}
}

@article{Wockel2008Principal2A,
url = {https://doi.org/10.1515/form.2011.020},
title = {Principal 2-bundles and their gauge 2-groups},
title = {},
author = {Christoph Wockel},
pages = {565--610},
volume = {23},
number = {3},
journal = {Forum Mathematicum},
doi = {doi:10.1515/form.2011.020},
year = {2011},
lastchecked = {2023-08-29}
}

@article{Garner:2023zqn,
    author = "Garner, Niklas and Williams, Brian R.",
    title = "{Raviolo vertex algebras}",
    eprint = "2308.04414",
    archivePrefix = "arXiv",
    primaryClass = "math.QA",
    month = "8",
    year = "2023"
}

@article{Larsson:2002pk,
    author = "Larsson, T. A.",
    title = "{Gerbes, covariant derivatives, $p$-form lattice gauge theory, and the Yang-Baxter equation}",
    eprint = "math-ph/0205017",
    archivePrefix = "arXiv",
    month = "5",
    year = "2002"
}

@article{Maillet:1989gg,
    author = "Maillet, Jean Michel and Nijhoff, Frank",
    title = "{Integrability for Multidimensional Lattice Models}",
    reportNumber = "CERN-TH-5332/89, INS-124/89",
    doi = "10.1016/0370-2693(89)91466-4",
    journal = "Phys. Lett. B",
    volume = "224",
    pages = "389--396",
    year = "1989"
}

@book{kosmann1997integrability,
  title={Integrability of Nonlinear Systems},
  author={Kosmann-Schwarzbach, Y. and Grammaticos, B. and Tamizhmani, K.M. and C.I.M.P.A. (Center)},
  isbn={9783540633532},
  lccn={97045868},
  series={International Association of Geodesy Symposia},
  url={https://books.google.ca/books?id=RpfvAAAAMAAJ},
  year={1997},
  publisher={Springer Berlin Heidelberg}
}

@article{KLIMCIK1995455,
title = {Dual non-Abelian duality and the Drinfeld double},
journal = {Physics Letters B},
volume = {351},
number = {4},
pages = {455-462},
year = {1995},
issn = {0370-2693},
doi = {https://doi.org/10.1016/0370-2693(95)00451-P},
url = {https://www.sciencedirect.com/science/article/pii/037026939500451P},
author = {C. Klimčík and P. Ševera}
}

@article{KLIMCIK1996116,
title = {Poisson-Lie T-duality},
journal = {Nuclear Physics B - Proceedings Supplements},
volume = {46},
number = {1},
pages = {116-121},
year = {1996},
issn = {0920-5632},
doi = {https://doi.org/10.1016/0920-5632(96)00013-8},
url = {https://www.sciencedirect.com/science/article/pii/0920563296000138},
author = {C. Klimčík}
}

@article{decoppet2022morita,
   title={The Morita Theory of Fusion 2-Categories},
   volume={7},
   url={http://dx.doi.org/10.21136/HS.2023.07},
   DOI={10.21136/hs.2023.07},
   number={1},
   journal={Higher Structures},
   publisher={Institute of Mathematics, Czech Academy of Sciences},
   author={Décoppet, Thibault D.},
   year={2023},
   month={May},
   pages={234–292} }

@book{neuchl1997representation,
  title={Representation Theory of Hopf Categories},
  author={Neuchl, M.},
  url={https://books.google.ca/books?id=gpgLHAAACAAJ},
  year={1997},
  publisher={München, Univ., Diss.}
}

@inbook{Kapranov:1994,
	title = {2-categories and Zamolodchikov tetrahedra equations},
	booktitle = {Algebraic groups and their generalizations: quantum and infinite-dimensional methods (University Park, PA, 1991)},
	series = {Proc. Sympos. Pure Math.},
	volume = {56},
	year = {1994},
	pages = {177{\textendash}259},
	publisher = {Amer. Math. Soc., Providence, RI},
	organization = {Amer. Math. Soc., Providence, RI},
	author = {Kapranov, M. M. and Voevodsky, V. A.}
}

@article{LU2002119,
title = {Inverses of 2 × 2 block matrices},
journal = {Computers \& Mathematics with Applications},
volume = {43},
number = {1},
pages = {119-129},
year = {2002},
issn = {0898-1221},
doi = {https://doi.org/10.1016/S0898-1221(01)00278-4},
url = {https://www.sciencedirect.com/science/article/pii/S0898122101002784},
author = {Tzon-Tzer Lu and Sheng-Hua Shiou}
}

@article{Khovanov:2000,
author = {Mikhail Khovanov},
title = {{A categorification of the Jones polynomial}},
volume = {101},
journal = {Duke Mathematical Journal},
number = {3},
publisher = {Duke University Press},
pages = {359 -- 426},
year = {2000},
doi = {10.1215/S0012-7094-00-10131-7},
URL = {https://doi.org/10.1215/S0012-7094-00-10131-7}
}

@article{Chen:2023tjf,
    author = "Chen, Hank and Girelli, Florian",
    title = "{Categorified Quantum Groups and Braided Monoidal 2-Categories}",
    eprint = "2304.07398",
    archivePrefix = "arXiv",
    primaryClass = "math.QA",
    month = "4",
    year = "2023"
}

@article{Chen:2022hct,
    author = "Chen, Hank and Girelli, Florian",
    title = "{Gauging the Gauge and Anomaly Resolution}",
    eprint = "2211.08549",
    archivePrefix = "arXiv",
    primaryClass = "hep-th",
    month = "11",
    year = "2022"
}

@article{Johnson_Freyd_2023,
   title={Minimal nondegenerate extensions},
   ISSN={1088-6834},
   url={http://dx.doi.org/10.1090/jams/1023},
   DOI={10.1090/jams/1023},
   journal={Journal of the American Mathematical Society},
   publisher={American Mathematical Society (AMS)},
   author={Johnson-Freyd, Theo and Reutter, David},
   year={2023},
   month={Jul} }

@article{GURSKI20114225,
title = {Loop spaces, and coherence for monoidal and braided monoidal bicategories},
journal = {Advances in Mathematics},
volume = {226},
number = {5},
pages = {4225-4265},
year = {2011},
issn = {0001-8708},
doi = {https://doi.org/10.1016/j.aim.2010.12.007},
url = {https://www.sciencedirect.com/science/article/pii/S0001870810004408},
author = {Nick Gurski}
}

@inproceedings{Rouquier2005CategorificationOS,
  title={Categorification of sl 2 and braid groups},
  author={Raphael Rouquier},
  year={2005}
}

@unpublished{rouquier:hal-00002981,
  TITLE = {{Categorification of the braid groups}},
  AUTHOR = {Rouquier, Raphael},
  URL = {https://hal.science/hal-00002981},
  NOTE = {working paper or preprint},
  YEAR = {2004},
  MONTH = Sep,
  HAL_ID = {hal-00002981},
  HAL_VERSION = {v1},
}

@article{Delcamp:2021szr,
    author = "Delcamp, Clement",
    title = "{Tensor network approach to electromagnetic duality in (3+1)d topological gauge models}",
    eprint = "2112.08324",
    archivePrefix = "arXiv",
    primaryClass = "cond-mat.str-el",
    doi = "10.1007/JHEP08(2022)149",
    journal = "JHEP",
    volume = "08",
    pages = "149",
    year = "2022"
}

@article{Delcamp:2023kew,
    author = "Delcamp, Clement and Tiwari, Apoorv",
    title = "{Higher categorical symmetries and gauging in two-dimensional spin systems}",
    eprint = "2301.01259",
    archivePrefix = "arXiv",
    primaryClass = "hep-th",
    doi = "10.21468/SciPostPhys.16.4.110",
    journal = "SciPost Phys.",
    volume = "16",
    number = "4",
    pages = "110",
    year = "2024"
}

@article{KongTianZhou:2020,
title = {The center of monoidal 2-categories in 3+1D Dijkgraaf-Witten theory},
journal = {Advances in Mathematics},
volume = {360},
pages = {106928},
year = {2020},
issn = {0001-8708},
doi = {https://doi.org/10.1016/j.aim.2019.106928},
url = {https://www.sciencedirect.com/science/article/pii/S0001870819305432},
author = {Liang Kong and Yin Tian and Shan Zhou}
}

@book{book-cft,
      author        = "Di Francesco, Philippe and Mathieu, Pierre and Sénéchal,
                       David",
      title         = "{Conformal field theory}",
      publisher     = "Springer",
      address       = "New York, NY",
      series        = "Graduate texts in contemporary physics",
      year          = "1997",
      url           = "https://cds.cern.ch/record/639405",
      doi           = "10.1007/978-1-4612-2256-9",
}

@article{Bobenko:1989,
author = {A. I. Bobenko and A. G. Reyman and M. A. Semenov-Tian-Shansky},
title = {{The Kowalewski top 99 years later: a Lax pair, generalizations and explicit solutions}},
volume = {122},
journal = {Communications in Mathematical Physics},
number = {2},
publisher = {Springer},
pages = {321 -- 354},
year = {1989},
doi = {cmp/1104178400},
URL = {https://doi.org/}
}

@article{Lax:1968fm,
    author = "Lax, P. D.",
    title = "{Integrals of Nonlinear Equations of Evolution and Solitary Waves}",
    journal = "Commun. Pure Appl. Math.",
    volume = "21",
    pages = "467--490",
    year = "1968"
}

@article{Severa:2005,
author = {Pavol {\v S}evera},
title = {Some title containing the words "homotopy" and "symplectic", e.g. this one},
journal = {Travaux Math{\'e}matiques},
volume = {16},
pages = {121-137},
year = 2005,
publisher = {Universit{\'e} du Luxembourg}

}

@article{Freidel:2019ees,
    author = "Freidel, Laurent and Livine, Etera R. and Pranzetti, Daniele",
    title = "{Gravitational edge modes: from Kac\textendash{}Moody charges to Poincar\'e networks}",
    eprint = "1906.07876",
    archivePrefix = "arXiv",
    primaryClass = "hep-th",
    doi = "10.1088/1361-6382/ab40fe",
    journal = "Class. Quant. Grav.",
    volume = "36",
    number = "19",
    pages = "195014",
    year = "2019"
}

@book{Wagemann+2021,
author = {Friedrich Wagemann},
doi = {doi:10.1515/9783110750959},
url = {https://doi.org/10.1515/9783110750959},
title = {Crossed Modules},
year = {2021},
publisher = {De Gruyter},
ISBN = {9783110750959},
lastchecked = {2022-06-28}
}

@article{Geiller:2020edh,
    author = "Geiller, Marc and Goeller, Christophe and Merino, Nelson",
    title = "{Most general theory of 3d gravity: Covariant phase space, dual diffeomorphisms, and more}",
    eprint = "2011.09873",
    archivePrefix = "arXiv",
    primaryClass = "hep-th",
    doi = "10.1007/JHEP02(2021)120",
    journal = "JHEP",
    volume = "02",
    pages = "120",
    year = "2021"
}

@article{Baez:2012,
	doi = {10.1090/s0065-9266-2012-00652-6},
  
	url = {https://doi.org/10.1090%2Fs0065-9266-2012-00652-6},
  
	year = 2012,
	publisher = {American Mathematical Society ({AMS})},
  
	volume = {219},
  
	number = {1032},
  
	author = {John Baez and Aristide Baratin and Laurent Freidel and Derek Wise},
  
	title = {Infinite-Dimensional Representations of 2-Groups},
  
	journal = {Memoirs of the American Mathematical Society}
}

@misc{Angulo:2018,
  doi = {10.48550/arxiv.1810.05740},
  
  url = {https://arxiv.org/abs/1810.05740},
  
  author = {Angulo, Camilo},
  
  title = {A cohomology theory for Lie 2-algebras and Lie 2-groups},
  
  publisher = {arXiv},
  
  year = {2018},
  
  copyright = {arXiv.org perpetual, non-exclusive license}
}

@article{Nikolaus_2014,
	doi = {10.1007/s40062-014-0077-4},
  
	url = {https://doi.org/10.1007%2Fs40062-014-0077-4},
  
	year = 2014,
	month = {mar},
  
	publisher = {Springer Science and Business Media {LLC}
},
  
	volume = {10},
  
	number = {3},
  
	pages = {565--622},
  
	author = {Thomas Nikolaus and Urs Schreiber and Danny Stevenson},
  
	title = {Principal $\infty$-bundles: presentations},
  
	journal = {Journal of Homotopy and Related Structures}
}

@ARTICLE{Adler:1978,
       author = {{Adler}, M.},
        title = "{On a Trace Functional for Formal Pseudo-Differential Operators and the Symplectic Structure of the Korteweg-Devries Type Equation.}",
      journal = {Inventiones Mathematicae},
         year = 1978,
        month = jan,
       volume = {50},
        pages = {219},
          doi = {10.1007/BF01410079},
       adsurl = {https://ui.adsabs.harvard.edu/abs/1978InMat..50..219A}
}

@article{Zhang:1991,
author = {M. Q. Zhang},
title = {{How to find the Lax pair from the Yang-Baxter equation}},
volume = {141},
journal = {Communications in Mathematical Physics},
number = {3},
publisher = {Springer},
pages = {523 -- 531},
year = {1991},
doi = {cmp/1104248391},
URL = {https://doi.org/}
}

@book{book-integrable, place={Cambridge}, series={Cambridge Monographs on Mathematical Physics}, title={Introduction to Classical Integrable Systems}, DOI={10.1017/CBO9780511535024}, publisher={Cambridge University Press}, author={Olivier Babelon, Denis Bernard and Michel Talon}, year={2003}, collection={Cambridge Monographs on Mathematical Physics}}

@article{Pfeiffer2007,
title = {2-Groups, trialgebras and their Hopf categories of representations},
journal = {Advances in Mathematics},
volume = {212},
number = {1},
pages = {62-108},
year = {2007},
issn = {0001-8708},
doi = {https://doi.org/10.1016/j.aim.2006.09.014},
url = {https://www.sciencedirect.com/science/article/pii/S0001870806003343},
author = {Hendryk Pfeiffer},
eprint = "0411468",
    archivePrefix = "arXiv",
    primaryClass = "math-ph",
}

@article{chen:2022,
  title = {Categorified Drinfel'd double and $BF$ theory: 2-groups in 4D},
  author = {Chen, Hank and Girelli, Florian},
  journal = {Phys. Rev. D},
  volume = {106},
  issue = {10},
  pages = {105017},
  numpages = {35},
  year = {2022},
  month = {Nov},
  publisher = {American Physical Society},
  doi = {10.1103/PhysRevD.106.105017},
  url = {https://link.aps.org/doi/10.1103/PhysRevD.106.105017}
}

@article{Mackaay:ek,
	author = {Marco Mackaay},
	eprint = {math/9903003},
	journal = {Adv. Math. 153(2) (2000), 353-390},
	title = {Finite groups, spherical 2-categories, and 4-manifold invariants},
	url = {https://arxiv.org/pdf/math/9903003.pdf},
	year={2000}}

@article{baez2004,
author = {Baez, John C. and Lauda, Aaron D.},
journal = {Theory and Applications of Categories [electronic only]},
language = {eng},
pages = {423-491},
publisher = {Mount Allison University, Department of Mathematics and Computer Science, Sackville},
title = {Higher-dimensional algebra. V: 2-Groups.},
url = {http://eudml.org/doc/124217},
volume = {12},
year = {2004},
}

@article{Baez:1995xq,
	archiveprefix = {arXiv},
	author = {Baez, J. C. and Dolan, J.},
	doi = {10.1063/1.531236},
	eprint = {q-alg/9503002},
	journal = {J. Math. Phys.},
	pages = {6073--6105},
	title = {{Higher dimensional algebra and topological quantum field theory}},
	volume = {36},
	year = {1995}}

@article{Baez:2004in,
	archiveprefix = {arXiv},
	author = {Baez, John and Schreiber, Urs},
	eprint = {hep-th/0412325},
	month = {12},
	title = {{Higher gauge theory: 2-connections on 2-bundles}},
	year = {2004}}

@article{Baez:2003fs,
	archiveprefix = {arXiv},
	author = {Baez, John C. and Crans, Alissa S.},
	eprint = {math/0307263},
	journal = {Theor. Appl. Categor.},
	pages = {492--528},
	title = {{Higher-Dimensional Algebra VI: Lie 2-Algebras}},
	volume = {12},
	year = {2004}}

@article{Crane:1994ty,
	archiveprefix = {arXiv},
	author = {Crane, Louis and Frenkel, Igor},
	doi = {10.1063/1.530746},
	eprint = {hep-th/9405183},
	journal = {J. Math. Phys.},
	pages = {5136--5154},
	title = {{Four-dimensional topological field theory, Hopf categories, and the canonical bases}},
	volume = {35},
	year = {1994}}

@article{Benini:2018reh,
	archiveprefix = {arXiv},
	author = {Benini, Francesco and C\'ordova, Clay and Hsin, Po-Shen},
	doi = {10.1007/JHEP03(2019)118},
	eprint = {1803.09336},
	journal = {JHEP},
	pages = {118},
	primaryclass = {hep-th},
	reportnumber = {SISSA 10/2018/FISI, SISSA-10-2018-FISI},
	title = {{On 2-Group Global Symmetries and their Anomalies}},
	volume = {03},
	year = {2019}}

@article{Meusburger:2021cxe,
	author = {Meusburger, Catherine},
	doi = {10.3390/sym13081324},
	journal = {Symmetry},
	number = {8},
	pages = {1324},
	title = {{Poisson\textendash{}Lie Groups and Gauge Theory}},
	volume = {13},
	year = {2021}}

@article{Chen:2013,
	author = {Chen, Zhuo and Sti{\'e}non, Mathieu and Xu, Ping},
	doi = {10.1016/j.geomphys.2013.01.006},
	issn = {0393-0440},
	journal = {Journal of Geometry and Physics},
	month = {Jun},
	pages = {59--68},
	publisher = {Elsevier BV},
	title = {Weak Lie 2-bialgebras},
	url = {http://dx.doi.org/10.1016/j.geomphys.2013.01.006},
	volume = {68},
	year = {2013}}

@article{Chen:2012gz,
	archiveprefix = {arXiv},
	author = {Chen, Zhuo and Sti{\'e}non, Mathieu and Xu, Ping},
	doi = {10.4310/jdg/1367438648},
	eprint = {1202.0079},
	journal = {J. Diff. Geom.},
	number = {2},
	pages = {209--240},
	primaryclass = {math.DG},
	title = {{Poisson 2-groups}},
	volume = {94},
	year = {2013}}

@article{Cordova:2018cvg,
	archiveprefix = {arXiv},
	author = {C\'ordova, Clay and Dumitrescu, Thomas T. and Intriligator, Kenneth},
	doi = {10.1007/JHEP02(2019)184},
	eprint = {1802.04790},
	journal = {JHEP},
	pages = {184},
	primaryclass = {hep-th},
	title = {{Exploring 2-Group Global Symmetries}},
	volume = {02},
	year = {2019}}

@article{Baez:2005sn,
author = {John C. Baez and Danny Stevenson and Alissa S. Crans and Urs Schreiber},
title = {{From loop groups to 2-groups}},
volume = {9},
journal = {Homology, Homotopy and Applications},
number = {2},
publisher = {International Press of Boston},
pages = {101 -- 135},
year = {2007},
}

@article{Bai_2013,
	author = {Chengming Bai and Yunhe Sheng and Chenchang Zhu},
	doi = {10.1007/s00220-013-1712-3},
	journal = {Communications in Mathematical Physics},
	month = {apr},
	number = {1},
	pages = {149--172},
	publisher = {Springer Science and Business Media {LLC}},
	title = {Lie 2-Bialgebras},
	url = {https://doi.org/10.1007%2Fs00220-013-1712-3},
	volume = {320},
	year = 2013}

@article{Babelon_1992,
	author = {Babelon, Olivier and Bernard, Denis},
	doi = {10.1007/bf02097626},
	issn = {1432-0916},
	journal = {Communications in Mathematical Physics},
	month = {Oct},
	number = {2},
	pages = {279--306},
	publisher = {Springer Science and Business Media LLC},
	title = {Dressing symmetries},
	url = {http://dx.doi.org/10.1007/BF02097626},
	volume = {149},
	year = {1992}}

@article{Semenov1992,
	author = {Semenov-Tyan-Shanskii, M. A.},
	day = {01},
	doi = {10.1007/BF01083527},
	issn = {1573-9333},
	journal = {Theoretical and Mathematical Physics},
	month = {Nov},
	note = {\url{https://doi.org/10.1007/BF01083527}},
	number = {2},
	pages = {1292--1307},
	title = {{Poisson-Lie groups. The quantum duality principle and the twisted quantum double}},
	volume = {93},
	year = {1992}}

\end{document}